\documentclass{article}

\usepackage{natbib}
\usepackage[margin=1in]{geometry}

\usepackage{amsthm}

\usepackage[utf8]{inputenc} 
\usepackage[T1]{fontenc}    
\usepackage{hyperref}       
\usepackage{url}            
\usepackage{booktabs}       
\usepackage{amsfonts}       
\usepackage{nicefrac}       
\usepackage{microtype}      
\usepackage{mathtools}
\usepackage{graphicx}

\usepackage{dsfont}

\theoremstyle{plain}
\newtheorem{theorem}{Theorem}
\newtheorem{lemma}[theorem]{Lemma}
\newtheorem{proposition}[theorem]{Proposition}
\newtheorem{corollary}[theorem]{Corollary}

\theoremstyle{definition}
\newtheorem{definition}[theorem]{Definition}
\newtheorem{example}[theorem]{Example} 

\usepackage[usenames,dvipsnames,svgnames,table]{xcolor}

\newcommand{\E}{\mathbb{E}}
\newcommand{\R}{\mathbb{R}}
\newcommand{\A}{\mathcal{A}}
\newcommand{\F}{\mathcal{F}}
\newcommand{\wt}{{w_i^{(t)}}}
\newcommand{\wtp}{{w_i^{(t+1)}}}
\newcommand{\wT}{{w_i^{(T)}}}
\newcommand{\mt}{{m_i^{(t)}}}
\newcommand{\Mt}{{M^{(t)}}}
\newcommand{\mT}{{m_i^{(T)}}}
\newcommand{\MT}{{M^{(T)}}}
\newcommand{\pt}{{p_i^{(t)}}}
\newcommand{\qt}{{q^{(t)}}}
\newcommand{\bt}{{b_i^{(t)}}}
\newcommand{\rt}{{r^{(t)}}}
\newcommand{\st}{{s_i^{(t)}}}
\newcommand{\Mpt}{{p^{(t)}}}
\newcommand{\Mbt}{{b^{(t)}}}

\newcommand{\Bern}{\text{Bern}}
\DeclareMathOperator*{\argmax}{argmax}
\newcommand{\Pt}{{\Phi^{(t)}}}
\newcommand{\Ptp}{{\Phi^{(t+1)}}}
\newcommand{\Ad}{\mathcal{A}_d}
\newcommand{\Ar}{\mathcal{A}_r}

\graphicspath{{figures/}}

\begin{document}

\title{Online Prediction with Selfish Experts}
\author{
	Tim Roughgarden\thanks{Department of Computer Science,
		Stanford University, 474 Gates Building, 353 Serra Mall, Stanford, CA 94305.
		This research was supported in part by NSF grant
		CCF-1215965.
		Email: {\tt tim@cs.stanford.edu}.} \and Okke Schrijvers\thanks{Department of Computer Science,
		Stanford University, 482 Gates Building, 353 Serra Mall, Stanford, CA 94305.
		This research was supported in part by NSF grant
		CCF-1215965.
		Email: {\tt okkes@cs.stanford.edu}.} 
}

	\maketitle

\begin{abstract}
We consider the problem of binary prediction with expert advice in
settings where experts have agency and seek to maximize their
credibility.
This paper makes three main contributions.  First, it defines a model
to reason formally about settings with selfish experts, and
demonstrates that  ``incentive compatible'' (IC) algorithms are
closely related to the design of proper scoring rules. 
Designing a good IC algorithm is easy if the designer's loss function
is quadratic, but for other loss functions, novel techniques are required.
Second, we design IC algorithms with good performance guarantees for
the absolute loss function. Third, we give a formal separation between
the power of online prediction with selfish experts and online
prediction with honest experts by proving lower bounds for both
IC and non-IC algorithms. 
In particular,
with selfish experts and the absolute loss function, there is no
(randomized) algorithm for online prediction---IC or otherwise---with asymptotically
vanishing regret.
\end{abstract}

\section{Introduction}
\label{s:intro}

In the months leading up to elections and referendums, a plethora of
pollsters try to figure out how the electorate is going to
vote. Different pollsters use different methodologies, reach different
people, and may have sources of random errors, so generally the polls
don't fully agree with each other. Aggregators such as Nate Silver's
FiveThirtyEight, The Upshot by the New York Times, and HuffPost
Pollster by the Huffington Post consolidate these different reports
into a single prediction, and hopefully reduce random
errors.\footnote{FiveThirtyEight: \url{https://fivethirtyeight.com/},
  The Upshot: \url{https://www.nytimes.com/section/upshot}, HuffPost
  Pollster: \url{http://elections.huffingtonpost.com/pollster}.}
FiveThirtyEight in particular has a solid track record for their
predictions, and as they are transparent about their methodology we
use them as a motivating example in this paper. To a first-order
approximation, they operate as follows: first they take the
predictions of all the different pollsters, then they assign a weight
to each of the pollsters based on past performance (and other
factors), and finally they use the weighted average of the pollsters
to run simulations and make their own prediction.\footnote{This is of course
a simplification. FiveThirtyEight also uses features like the
  change in a poll over time, the state of the economy, and
  correlations between states.  See
  \url{https://fivethirtyeight.com/features/how-fivethirtyeight-calculates-pollster-ratings/}
  for details. Our
  goal in this paper is not to accurately model all of the fine
  details of FiveThirtyEight (which are anyways changing all the time).
Rather, it is to formulate a general model of
  prediction with experts that clearly illustrates why incentives matter.}

But could the presence of an institution that rates pollsters
inadvertently create perverse incentives for the pollsters? The
FiveThirtyEight pollster ratings are publicly
available.\footnote{\url{https://projects.fivethirtyeight.com/pollster-ratings/}}
The ratings can be interpreted as a reputation, and a low rating can
negatively impact future revenue opportunities for a
pollster. Moreover, it has been demonstrated in practice that
experts do not always report their true beliefs about
future events. For example, in weather forecasting there is a known
``wet bias,'' where consumer-facing weather forecasters deliberately
\emph{over}estimate low chances of rain (e.g. a $5\%$ chance of rain
is reported as a $25\%$ chance of rain) because people don't like to
be surprised by rain \citep{BK08}.

These examples motivate the development of models of aggregating
predictions that endow agency to the data sources.\footnote{More
  generally, one can investigate how the presence of machine learning
  algorithms affects data generating processes, either during
  learning, e.g. \citep{DFP10,CDP15}, or during deployment,
  e.g. \citep{HMPW16,BS11}. We discuss some of this work in the related
  work section.}  While there are multiple models in which we can
investigate this issue, a natural candidate is the problem of
prediction with expert advice. 
By focusing on a standard model, we abstract away from the fine details of FiveThirtyEight (which are anyways changing all the time), which allows us to formulate a general model of prediction with experts that clearly illustrates why incentives matter.
 In the classical
model~\citep{LW94,FS97}, at each time step, several experts make
predictions about an unknown event.  An online prediction algorithm
aggregates experts' opinions and makes its own prediction at each time
step.  After this prediction, the event at this time step is realized
and the algorithm incurs a loss as a function of its prediction and
the realization. To compare its performance against individual
experts, for each expert the algorithm calculates what its loss would
have been had it always followed the expert's prediction.
While the problems introduced in this paper are relevant for general
online prediction, to focus on the most interesting issues we
concentrate on the case of binary events, and real-valued predictions
in $[0,1]$. For different applications, different notions of loss are
appropriate, so we parameterize the model by a loss function
$\ell$. Thus our formal model is: at each time step $t=1,2,\ldots,T$:
\begin{enumerate}

\item Each expert $i$ makes a prediction $\pt \in [0,1]$, with higher
  values indicating stronger advocacy for the event ``1.''

\item The online algorithm commits to a probability distribution
  over $\{0,1\}$, with $\qt$ denoting the probability assigned to~``1.''

\item The outcome $\rt \in \{0,1\}$ is realized.

\item The algorithm incurs a  
loss of $\ell(\qt,\rt)$ and calculates for each
  expert~$i$  a loss of $\ell(\pt,\rt)$.

\end{enumerate}

The standard goal in this problem is to design an online prediction
algorithm that is guaranteed to have expected loss not much larger
than that incurred by the best expert in hindsight.
The classical solutions
maintain a weight for each expert and make a prediction
according to which outcome has more expert weight behind it.
An expert's weight can be interpreted as a measure of its credibility
in light of its past performance.
The (deterministic) Weighted Majority (WM) algorithm always chooses
the outcome with more expert weight.  The Randomized Weighted Majority
(RWM) algorithm randomizes between the two outcomes with probability
proportional to their total expert weights.  The most common method of
updating experts' weights is via multiplication by $1-\eta \ell(p_i^{(t)}, r^{(t)})$
after each time step~$t$, where $\eta$ is the learning rate.  We call
this the ``standard'' or ``classical'' version of the WM and RWM algorithm.

The classical model instills no agency in the experts.
To account for this, in this
paper we replace Step 1 of the classical model by:
\begin{itemize}

\item [1a.] Each expert~$i$ formulates a belief $\bt \in [0,1]$.

\item [1b.] Each expert~$i$ reports a prediction $\pt \in [0,1]$ to
  the algorithm.

\end{itemize}
Each expert now has two types of loss at each time step --- the {\em
  reported loss}
$\ell(\pt,\rt)$ with respect to the reported prediction and the {\em
  true loss}
$\ell(\bt,\rt)$ with respect to her true beliefs.\footnote{When we speak
  of the best expert in hindsight, we are always referring to the
  true losses. Guarantees with respect to reported losses follow from
  standard results~\citep{LW94,FS97,CBMS07}, but are not immediately meaningful.}

When experts care about the weight that they are assigned, and with it
their reputation and influence in the algorithm,
different loss functions can lead to different expert behaviors. 
For example, in Section~\ref{s:model} we observe that for the quadratic loss
function, in the standard WM and RWM algorithms,
experts have no reason to misreport their beliefs. 
The next example shows that this is not the case for 
other loss functions, such as the absolute loss function.\footnote{The
  loss function is often tied to the particular application.  For example,
  in the current FiveThirtyEight pollster rankings, the performance of a
  pollster is primarily measured according to an absolute loss
  function and also whether the
  candidate with the highest polling numbers ended up winning (see
  \url{https://github.com/fivethirtyeight/data/tree/master/pollster-ratings}). However,
  in 2008 FiveThirtyEight used the notion of ``pollster introduced
  error'' or PIE, which is the square root of a difference of squares,
  as the most important feature in calculating the weights, see
  \url{https://fivethirtyeight.com/features/pollster-ratings-v31/}.}

\begin{example}\label{ex:obvious}
Consider the standard WM algorithm, where each expert initially has
unit weight, and an expert's weight is multiplied by $1-\eta |\pt - \rt|$ at
a time step $t$, where $\eta \in
(0,\tfrac{1}{2})$ is the learning rate.
Suppose there are two experts and $T=1$, and that $b^{(1)}_1 = .49$
while $b^{(1)}_2 = 1$.  Each expert reports to maximize her expected
weight.
Expanding, for each $i=1,2$ we have
\begin{align*}
\E[w^{(1)}_i] &= \Pr(r^{(1)} = 1)\cdot (1-\eta(1-p_i^{(1)})) +
                \Pr(r^{(1)} =  0)\cdot (1-\eta p_i^{(1)})\\
&= b_i^{(1)}\cdot (1-\eta (1-p_i^{(1)})) + (1-b_i^{(1)})\cdot (1-\eta p_i^{(1)})\\
&= b_i^{(1)} - b_i^{(1)}\eta + b_i^{(1)}\eta p_i^{(1)} + 1 - \eta p^{(1)}_i - b_i^{(1)}+b_i^{(1)}\eta p_i^{(1)}\\
&= 2b_i^{(1)} \eta p_i^{(1)}  - p_i^{(1)}\eta - b_i^{(1)}\eta + 1,
\end{align*}

where all expectations and probabilities are with respect to the true
beliefs of agent~$i$.
To maximize this expected weight over the possible reports $p_i^{(1)}
\in [0,1]$, we can omit the second two terms (which are independent of
$p^{(1)}_i$) and divide out by $\eta$ to obtain
	\begin{align*}
	\argmax_{p_i^{(1)} \in [0,1]} 2b_i^{(1)} \eta p_i^{(1)}  -
          p_i^{(1)}\eta - b_i^{(1)}\eta + 1&= \argmax_{p_i^{(1)} \in [0,1]} p_i^{(1)}(2b_i^{(1)} - 1)\\
	&= \begin{cases}
	1 & \text{if } b_i^{(1)} \geq \frac{1}{2} \\
	0 & \text{otherwise.}
	\end{cases}
	\end{align*}
Thus an expert always reports a point mass on whichever
outcome she believes more likely.  In our example, the second expert
will report her true beliefs ($p^{(t)}_2 = 1$) while the first will
not ($p^{(t)}_1 = 0$).
While the combined true beliefs of the experts clearly favor outcome
$1$, the WM algorithm sees two opposing predictions and must break
ties arbitrarily between them. 
\end{example}

This conclusion can also be drawn directly from the property elicitation literature. Here, the absolute loss function is known to elicit the median \citep{B76}\citep{T79}, and since we have binary realizations, the median is either 0 or 1. Example~\ref{ex:obvious} shows that for the absolute loss function the standard WM algorithm is not
``incentive-compatible'' in a sense that we formalize in
Section~\ref{s:model}.  There are similar examples for
the other commonly studied weight update rules and for the RWM
algorithm. We might care about truthful reporting for its own sake,
but additionally the worry is that non-truthful reports will impede
our ability to get good regret guarantees (with respect to experts'
true losses).

We study several fundamental questions about online prediction with selfish experts:
\begin{enumerate}

\item What is the design space of ``incentive-compatible'' online
  prediction  algorithms, where every expert is incentivized to report
  her true beliefs?

\item Given a loss function like absolute loss,
are there incentive-compatible algorithm that obtain good regret guarantees?

\item 
Is online prediction with selfish experts strictly harder
than in the classical model with honest experts?

\end{enumerate}

\subsection{Our Results}
The first contribution of this paper is the development of a model for
reasoning formally about the design and analysis
of weight-based online prediction
algorithms when experts are selfish (Section~\ref{s:model}), and the
definition of an ``incentive-compatible'' (IC) such algorithm.
Intuitively, an IC algorithm is such that each expert wants to report
its true belief at each time step.
We demonstrate that the design of IC online
prediction algorithms is closely related to the design of strictly
proper scoring rules. Using this, we show that for the quadratic loss
function, the standard WM and RWM
algorithms are IC, whereas these algorithms are not
generally IC for other loss functions.

Our second contribution is the design of IC
prediction algorithms for the absolute loss function with non-trivial performance guarantees.  For
example, our best result for deterministic algorithms is: the WM
algorithm, with experts' weights evolving according to the spherical
proper scoring rule (see Section~\ref{s:det-ub}), is IC
and has loss at most $2+\sqrt{2}$ times the loss of best
expert in hindsight (in the limit as $T \rightarrow \infty$).  A
variant of the RWM algorithm with the Brier scoring rule is IC and has expected loss at most 2.62 times that of
the best expert in hindsight (also in the limit, see Section~\ref{s:rand}).

Our third and most technical contribution is a
formal separation between online prediction with selfish experts and the traditional setting with honest experts.
Recall that with honest experts, the classical (deterministic) WM
algorithm has loss at most twice that of the best expert in hindsight
(as $T \rightarrow \infty$) \cite{LW94}.  We prove that the worst-case loss of
every (deterministic) IC algorithm (Section~\ref{s:det-lb}) and every
non-IC algorithm satisfying mild technical conditions
(Section~\ref{s:det-non-ic-lb}) has worst-case loss bounded away from
twice that of the best expert in hindsight (even as
$T \rightarrow \infty$).  A consequence of our lower bound is that,
with selfish experts, there is no natural (randomized) algorithm for
online prediction---IC or otherwise---with asymptotically vanishing
regret.

\subsection{Related Work}
We believe that our model of online prediction over time with selfish
experts is novel.  We next survey the multiple other ways in which online
learning and incentive issues have been blended, and the other efforts 
to model incentive issues in machine learning.

There is a large literature on prediction and decision markets
(e.g.~\citep{chen_pennock10,horn14}), which also aim to aggregate
information over time from multiple parties and make use of proper
scoring rules to do it. There are several major differences between
our model and prediction markets.  First, in our model, the goal is to
predict a sequence of events, and there is feedback (i.e., the
realization) after each one.  In a prediction market, the goal is to
aggregate information about a single event, with feedback provided
only at the end (subject to secondary objectives, like bounded
loss).\footnote{In the even more distantly related peer prediction
scenario~\citep{MRZ05}, there is never any realization at all.}
Second, our goal is to make accurate predictions, while that of a
prediction market is to aggregate information.  For example, if
one expert is consistently incorrect over time, we would like to
ignore her reports rather than aggregate them with others' reports.
Finally, while there are strong mathematical connections between cost
function-based prediction markets and regularization-based online
learning algorithms in the standard (non-IC) model~\citep{ACV12}, there
does not appear to be analogous connections with online
prediction with selfish experts.

There is also an emerging literature on ``incentivizing exploration''
(as opposed to exploitation) in partial feedback models such as the
bandit model (e.g.~\citep{FKKK14,MSSW16}).  Here, the incentive issues
concern the learning algorithm itself, rather than the experts (or
``arms'') that it makes use of.

The question of how an expert should report beliefs has been studied before in the literature on strictly proper scoring rules \citep{B50,M56,S71,GR07}, but this literature typically considers the evaluation of a single prediction, rather than low-regret learning. The work by \citet{BD89} specifically looks at the question of how an expert should respond to an aggregator who assigns and updates weights based on their predictions. Their work focuses on optimizing relative weight under different objectives and informational assumptions. However, it predates the work on low-regret learning, and it does not include performance guarantees for the aggregator over time. \citet{B12} discusses a model in which an aggregator wants to take a specific action based on predictions that she elicits from experts. He explores incentive issues where experts have a stake in the action that is taken by the decision maker.

Finally, there are many works that fall under the broader umbrella of
incentives in machine learning.
Roughly, work in this area can be divided into two genres: incentives
during the learning stage, or incentives during the deployment
stage. During the learning stage, one of the main considerations is
incentivizing data providers to exert effort to generate high-quality
data. There are several recent works that propose ways to elicit data
in crowdsourcing applications in repeated settings through payments,
e.g. \citep{CDP15,SZ15,LC16}. Outside of crowdsourcing,
\citet{DFP10} consider a regression task where different experts have
their own private data set, and they seek to influence the learner to
learn a function such that the loss of their private data set with
respect to the function is low.

During deployment, the concern is that the input is given by agents who have a stake in the result of the classification, e.g. an email spammer wishes to avoid its emails being classified as spam. \citet{BS11} model a learning task as a Stackelberg game. On the other hand \citet{HMPW16} consider a cost to changing data, e.g. improving your credit score by opening more lines of credit, and give results with respect to different cost functions.

Online learning does not fall neatly into either learning or
deployment, as the learning is happening while the system is
deployed. \citet{BKS10} consider the problem of no-regret learning
with selfish experts in an ad auction setting, where the incentives
come from the allocations and payments of the auction, rather than
from weights as in our case.

\subsection{Organization}
Section~\ref{s:model} formally defines weight-update online prediction
algorithms and shows a connection between algorithms that are incentive compatible and proper scoring
rules. We use the formalization to show that when we care about achieving guarantees for quadratic losses, the standard WM and RWM algorithms work well. Since the standard algorithm fails to work well for absolute losses, we focus in the remainder of the paper on proving guarantees for this case.

Section~\ref{s:det-ub} gives a deterministic weight-update
online prediction algorithm that is incentive-compatible and has absolute loss
at most $2+\sqrt 2$ times that of the best expert in hindsight (in the
limit).  Additionally we show that the weighted majority algorithm with
the standard update rule has a worst-case true loss of at least $4$ times the best expert in hindsight. 

To show the limitations of online prediction with selfish experts, we break our lower bound results into two parts. In Section~\ref{s:det-lb} we show that any deterministic incentive compatible weight-update online prediction algorithm has worst case loss bounded away from $2$, even as $T\rightarrow \infty$. Then in Section~\ref{s:det-non-ic-lb} we show that under mild technical conditions, the same is true for non-IC algorithms. 

Section~\ref{s:rand} contains our results for randomized algorithms. It shows that the lower bounds for deterministic algorithms imply that under the same conditions randomized algorithms cannot have asymptotically vanishing regret. We do give an IC randomized algorithm that achieves worst-case loss at most 2.62 times that of the best expert in hindsight (in the limit).

Finally, in Section~\ref{s:simulations} we show  simulations that indicate that different IC methods show similar regret behavior, and that their regret is substantially better than that of the non-IC standard algorithms, suggesting that the worst-case characterization we prove holds more generally.

The appendix contains omitted proofs (Appendix~\ref{s:proofs}),
and a discussion on the selecting appropriate proper scoring rules for good
guarantees (Appendix~\ref{s:selection}).

\section{Preliminaries and Model}
\label{s:model}

\subsection{Standard Model}
At each time step $t\in 1, ..., T$ we want to predict a binary
realization $\rt\in\{0,1\}$. To help in the prediction, we have access
to $n$ experts that for each time step report a prediction
$\pt \in [0,1]$ about the realization. The realizations are determined
by an oblivious adversary, and the predictions of the experts may
or may not be accurate. The goal is to use the predictions of the
experts in such a way that the algorithm performs nearly as well as
the best expert in hindsight. 
Most of the algorithms proposed for this problem
fall into the following framework.

\begin{definition}[Weight-update Online Prediction Algorithm]
	\label{def:wuopa}
	A weight-update online prediction algorithm maintains a weight
        $\wt$ for each expert and makes its prediction $\qt$ based
        on $\sum_{i=1}^n \wt\pt$ and $\sum_{i}^n \wt(1-\pt)$. After the
        algorithm makes its prediction, the realization $\rt$ is
        revealed, and the algorithm updates the weights of experts
        using the rule
	\begin{equation}
		\wtp = f\left(\pt, \rt\right)\cdot \wt,
	\end{equation}
	where $f : [0,1]\times \{0,1\} \rightarrow \R^+$ is a positive function on its domain.
\end{definition}

The standard WM algorithm has $f(\pt,\rt)=1-\eta\ell(\pt,\rt)$ where
 $\eta\in(0,\tfrac12)$ is the learning rate, and
predicts $q^{(t)}=1$ if and only if
$\sum_{i}^n \wt\pt\geq\sum_{i}^n \wt(1-\pt)$. Let the total loss of
the algorithm be $\MT = \sum_{t=1}^T \ell(\qt,\rt)$ and let the
total loss of expert $i$ be $\mT = \sum_{t=1}^T \ell(\pt, \rt)$. The MW
algorithm has the property that
$\MT \leq 2(1+\eta)\mT + \frac{2\ln n}{\eta}$ for each expert $i$, and
RWM ---where the algorithm picks $1$ with probability proportional to
$\sum_{i}^n \wt\pt$--- satisfies $\MT \leq (1+\eta)\mT + \frac{\ln n}{\eta}$
for each expert $i$ \citep{LW94}\citep{FS97}.

The notion of \emph{``no $\alpha$-regret''} \citep{KKL09} captures the idea that the per time-step loss of an algorithm is $\alpha$ times that of the best expert in hindsight, plus a term that goes to $0$ as $T$ grows:

\begin{definition}[$\alpha$-regret]
	An algorithm is said to have \emph{no $\alpha$-regret} if $\MT \leq\alpha\min_i\mT + o(T).$
\end{definition}
By taking $\eta=O(1/\sqrt T)$, MW is a no $2$-regret algorithm, and RWM is a no $1$-regret algorithm.

\subsection{Selfish Model} 
We consider a model in which experts have agency about the prediction
they report, and care about the weight that they are assigned. In the
selfish model, at time $t$ the expert formulates a private belief
$\bt$ about the realization, but she is free to report any prediction
$\pt$ to the algorithm. Let $\text{Bern}(p)$ be a Bernoulli random
variable with parameter $p$. For any non-negative weight update
function $f$,
$$\max_p \E_\bt[\wtp]\ =\ \max_p \E_{r\sim \Bern\left(\bt\right)}[f\left(p,r\right)\wt]\ =\ \wt\cdot\left(\max_p  \E_{r\sim\Bern \left(\bt\right)}[f\left(p,r\right)]\right).$$
So expert $i$ will report whichever $\pt$ will maximize the expectation of the weight update function.

Performance of an algorithm with respect to the reported loss of experts follows from the standard analysis \citep{LW94}. However, the true loss may be worse (in Section~\ref{s:det-ub} we show this for the standard update rule, Section~\ref{s:det-non-ic-lb} shows it more generally). Unless explicitly stated otherwise, in the remainder of this paper $\mT = \sum_{t=1}^T \ell(\bt, \rt)$ refers to the \emph{true} loss of expert $i$. For now this motivates restricting the weight update rule $f$ to functions where reporting $\pt = \bt$ maximizes the expected weight of experts. We call these weight-update rules \emph{Incentive Compatible (IC)}.

\begin{definition}[Incentive Compatibility]
	\label{def:ic}
        A weight-update function $f$ is {\em incentive compatible
          (IC)} if reporting the true belief $\bt$ is 
always a best response 
        for every expert at every time step. 
It is {\em strictly IC} when
        $\pt = \bt$ is the only best response.
\end{definition}
By a ``best response,'' we mean an expected utility-maximizing report,
where the expectation is with respect to the expert's beliefs.

\paragraph{Collusion.} The definition of IC does not rule out the possibility that experts can collude to jointly misreport to improve their weights. We therefore also consider a stronger notion of incentive compatibility for groups with transferable utility.

\begin{definition}[Incentive Compatibility for Groups with Transferable Utility]
	\label{def:gic}
	A weight-update function $f$ is {\em incentive compatible for groups with transferable utility (TU-GIC)} if for every subset $S$ of players, the total expected weight of the group $\sum_{i\in S} \E_\bt[w_i^{(t+1)}]$ is maximized by each reporting their private belief $\bt$.
\end{definition}

Note that TU-GIC is a strictly stronger concept than IC, as for any algorithm that is TU-GIC, the condition needs to hold for groups of size 1, which is the definition of IC. The concept is also strictly stronger than that of GIC with nontransferable utility (NTU-GIC), where for every group $S$ it only needs to hold that there are no alternative reports that would make no member worse off, and at least one member better off \citep{M99}\citep{JM07}.

\subsection{Proper Scoring Rules} Incentivizing truthful reporting of beliefs has been studied extensively, and the set of functions that do this is called the set of proper scoring rules. Since we focus on predicting a binary event, we restrict our attention to this class of functions.

\begin{definition}[Binary Proper Scoring Rule, \citep{S89}]
	\label{def:psr}
A function $f : [0,1]\times\{0,1\}\rightarrow \R\cup \{\pm\infty\}$ is a \emph{binary proper scoring rule} if it is finite except possibly on its boundary and whenever for $p \in [0,1]$
	$$p \in \max_{q\in[0,1]} p\cdot f(q, 1) + (1-p)\cdot f(q, 0).$$
\end{definition}
A function $f$ is a \emph{strictly} proper scoring rule if $p$ is the only value that maximizes the expectation. The first perhaps most well-known proper scoring rule is the Brier scoring rule.
\begin{example}[Brier Scoring Rule, \citep{B50}]
	The \emph{Brier score} is $Br(p,r) = 2p_r - (p^2 + (1-p)^2)$ where $p_r = pr + (1-p)(1-r)$ is the report for the event that materialized.
\end{example}
We will use the Brier scoring rule in Section~\ref{s:rand} to construct an incentive-compatible randomized algorithm with good guarantees.
 The following proposition follows directly from Definitions~\ref{def:ic}~and~\ref{def:psr}:

\begin{proposition}
	\label{prop:psr}
A weight-update rule $f$ is (strictly) incentive compatible if and only if $f$ is a (strictly) proper scoring rule.
\end{proposition}

Surprisingly, this result remains true even when experts can collude.
While the realizations are obviously correlated, linearity of expectation causes the sum to be maximized exactly when each expert maximizes their expected weight.

\begin{proposition}
	\label{prop:gic}
	A weight-update rule $f$ is (strictly) incentive compatible for groups with transferable utility if and only if $f$ is a (strictly) proper scoring rule.
\end{proposition}

Thus, for online prediction with selfish experts, we get TU-GIC ``for free.''It is quite uncommon for problems in non-cooperate game theory to admit good TU-GIC solutions. For example, results for auctions (either for revenue or welfare) break down once bidders collude, see \citep{GH05} and references therein for more examples from theory and practice. In the remainder of the paper we will simply us IC to refer to both incentive compatibility and incentive compatibility for groups with transferable utility, as strictly proper scoring rules lead to algorithms that satisfy both definitions.

So when considering incentive compatibility in the online prediction
with selfish experts setting, we are restricted to considering proper
scoring rules as weight-update rules. Moreover, since $f$ needs to be
positive, only bounded proper scoring rules can be used. 
Conversely, any bounded scoring
rule can be used, possibly after an affine transformation (which preserves
proper-ness). 
Are there any proper scoring rules that give an online prediction
algorithm with a good performance guarantee?

\subsection{Online Learning with Quadratic Losses}
The first goal of this paper is to describe the class of algorithms that lead incentive compatible learning. Proposition~\ref{prop:psr} answers this question, so we can move over to our second goal, which is for different loss functions, do there exist incentive compatible algorithms with good performance guarantees? In this subsection we show that a corollary of Proposition~\ref{prop:psr} is that the standard MW algorithm with the quadratic loss function $\ell(p,r) = (p-r)^2$ is incentive compatible.

\begin{corollary}
	The standard WM algorithm with quadratic losses, i.e. $w_i^{(t+1)}=(1-\eta(\pt-\rt))^2\cdot\wt$ is incentive compatible.
\end{corollary}
\begin{proof}
	By Proposition~\ref{prop:psr} it is sufficient to show that $\bt=\max_p \bt\cdot (1-\eta(p-1)^2) + (1-\bt)\cdot (p-0)^2$.
	
\begin{align*}
&\max_p \bt\cdot (1-\eta(p-1)^2) + (1-\bt)\cdot (1-\eta(p-0)^2)\\
= &\max_p\bt - \bt\eta p^2 + 2\bt\eta p - \bt\eta + 1-\bt - \eta p^2 + \bt\eta p^2\\
= &\max_p 1- \bt\eta  + 2\bt\eta p  - \eta p^2 \\
= &\max_p 1- \bt\eta  + \eta p(2\bt-p)
\end{align*}
To solve this for $p$, we take the derivative with respect to $p$:
$\frac{d}{dp} 1- \bt\eta  + \eta p(2\bt-p) = \eta(2\bt - 2p)$. So the maximum expected value is uniquely $p=\bt$.
\end{proof}
A different way of proving the Corollary is by showing that the
standard update rule with quadratic losses can be translated into the
Brier strictly proper scoring rule. Either way, for applications with
quadratic losses, 
the standard algorithm already works out of the box.
However, as we saw in Example~\ref{ex:obvious}, 
this is not the case with the absolute loss function. 
As the absolute loss function arises in practice---recall
that FiveThirtyEight uses absolute loss to calculate their pollster
ratings---in the remainder of this paper we focus on answering
questions (2) and (3) from the introduction for the absolute loss function.


\section{Deterministic Algorithms for Selfish Experts}
\label{s:det-ub}

This section studies the question if there are good online prediction
algorithms for the absolute loss function. We restrict our attention
here to deterministic algorithms; Section~\ref{s:rand} gives a
randomized algorithm with good guarantees.

Proposition~\ref{prop:psr} tells us that for selfish experts to have
a strict incentive to report truthfully, the weight-update rule must
be a strictly proper scoring rule. 
This section gives a
deterministic algorithm based on the \emph{spherical} strictly proper
scoring rule that has no $(2+\sqrt 2)$-regret
(Theorem~\ref{thm:spherical-det-ub}). Additionally, we consider the
question if the non-truthful reports from experts in using the
standard (non-IC) WM algorithm are harmful. We show that this is the
case by proving that the algorithm is not a  no $(4-O(1))$-regret algorithm, for any constant smaller than 4 (Proposition~\ref{prop:standard-det-lb}). This shows that,
when experts are selfish, the
IC online prediction algorithm
with the spherical rule outperforms
the standard WM algorithm (in the worst case).

\subsection{Deterministic Online Prediction using a Spherical Rule}
\label{ss:det-ub}
We next give an algorithm that uses a strictly proper scoring rule
that is based on the spherical rule scoring rule.\footnote{See Appendix~\ref{s:selection} for intuition about why this rule
yields better results than other natural candidates, such as the Brier
scoring rule.} In the following, let $\st = |\pt - \rt|$ be the absolute loss of expert $i$.

Consider the following weight-update rule:
\begin{align}
\label{eq:spherical}
f_{sp}\left(\pt, \rt\right) =  1-\eta\left(1-\frac{1-\st}{\sqrt{\left(\pt\right)^2 + \left(1-\pt\right)^2}}\right).
\end{align}
The following proposition establishes that this is in
fact a strictly proper scoring rule.
\begin{proposition}
The spherical weight-update rule in \eqref{eq:spherical} is a strictly
proper scoring rule.	
\end{proposition}
\begin{proof}
  The standard spherical strictly proper scoring rule is
  $(1-\st)/\sqrt{(\pt)^2 + (1-\pt)^2}$. Any positive affine
  transformation of a strictly proper scoring rule yields another
  strictly proper scoring rule, see e.g. \citep{GR07}, hence
  $(1-\st)/\sqrt{(\pt)^2 + (1-\pt)^2}-1$ is also a strictly proper
  scoring rule. Now we multiply this by $\eta$ and add $1$ to
  obtain$$1+\eta\left(\frac{1-\st}{\sqrt{\left(\pt\right)^2 +
        \left(1-\pt\right)^2}}-1\right),$$
  and rewriting proves the claim.
\end{proof}

In addition to incentivizing truthful reporting,
the WM algorithm with the update rule $f_\text{sp}$ 
does not do much worse than the best expert in hindsight. 
(See the appendix for the proof.)

\begin{theorem}\label{thm:spherical-det-ub}
The WM algorithm with weight-update rule \eqref{eq:spherical} for $\eta=O(1/\sqrt T)<\frac12$ has no $(2+\sqrt 2)$-regret.
\end{theorem}

\subsection{True Loss of the Non-IC Standard Rule}
\label{ss:true-loss}

It is instructive to compare the guarantee in
Theorem~\ref{thm:spherical-det-ub} with the performance of the
standard (non-IC) WM algorithm.
With the standard weight update function $f(\pt,\rt) = 1-\eta\st$
for $\eta\in(0,\tfrac12)$, the WM algorithm has the guarantee that
$\MT\leq 2\left((1+\eta)\mT + \tfrac{\ln n}{\eta}\right)$
with respect to the {\em reported} loss of experts. However, 
Example~\ref{ex:obvious} demonstrates that 
this algorithm incentivizes extremal reports, i.e. if $\bt\in [0,\tfrac12)$
the expert will report $\pt=0$ and if $\bt\in(\tfrac12,1]$ the expert
will report $1$. The following proposition shows that, in the worst
case, this algorithm does no better than a factor $4$ times the {\em
  true} loss of the best expert in hindsight.
Theorem~\ref{thm:spherical-det-ub} shows that a suitable IC algorithm
can obtain a superior worst-case guarantee.

\begin{proposition}
	\label{prop:standard-det-lb}
The standard WM algorithm with weight-update rule $f\left(\pt,\rt\right) = 1-\eta|\pt - \rt|$ results in a total worst-case loss no better than
	$$
	\MT \geq 4\cdot \min_i \mT - o(1).
	$$
\end{proposition}

\begin{proof}
  Let $A$ be the standard weighted majority algorithm. We create an
  instance with 2 experts where $\MT \geq 4\cdot \min_i \mT -
  o(1)$.
  Let the reports $p_1^{(t)}=0$, and $p_2^{(t)}=1$ for all
  $t\in1, ..., T$; we will define $\bt$ shortly. Given the reports,
  $A$ will choose a sequence of predictions, let $\rt$ be $1$ whenever
  the algorithm chooses $0$ and vice versa, so that $\MT=T$.
	
Now for all $t$ such that $\rt = 1$, set $b_1^{(t)}=\tfrac12 -
\epsilon$ and $b_2^{(t)}=1$, and for all $t$ such that $\rt=0$ set
$b_1^{(t)}=0$ and $b_2^{(t)}=\tfrac12 + \epsilon$, for small
$\epsilon>0$. Note that the beliefs $\bt$ indeed lead to the reports
$\pt$ since $A$ incentivizes rounding the reports to the nearest
integer. 
	
Since the experts reported opposite outcomes, their combined total
number of incorrect reports is $T$, hence the best expert had a
reported loss  of at most $T/2$. For each incorrect report $\pt$, the
real loss of expert is $|\rt - \bt| = \tfrac12+\epsilon$, hence
$\min_i \mT \leq \left(\tfrac12 + \epsilon\right)T/2$, while $\MT =
T$. Taking $\epsilon = o(T^{-1})$ yields the claim. 
\end{proof}

\section{The Cost of Selfish Experts for IC algorithms}
\label{s:det-lb}

We now address the third fundamental question:
whether or not online prediction with selfish experts is strictly
harder than with honest experts.
In this section we restrict our attention to deterministic algorithms; we
extend the results to randomized algorithms in
Section~\ref{s:rand}. As there exists a deterministic algorithm for
honest experts with no $2$-regret, showing a separation between honest and
selfish experts boils down to proving that there exists a constant
$\delta$ such that the worst-case loss is no better than a factor of
$2+\delta$
(with $\delta$ bounded away from 0 as $T\rightarrow \infty$).

In this
section we show that such a $\delta$ exists for all incentive
compatible algorithms, and that $\delta$ depends on properties of a
``normalized'' version of the weight-update rule $f$, independent of
the learning rate. This implies that the lower bound also holds for
algorithms that, like the classical prediction algorithms, use a
time-varying learning rate. In Section~\ref{s:det-non-ic-lb} we show
that under mild technical conditions the true loss of non-IC
algorithms is also bounded away from $2$, and in Section~\ref{s:rand}
the lower bounds for deterministic algorithms are used to show that
there is no randomized algorithm that achieves vanishing regret.

To prove the lower bound, we have to be specific about which set of
algorithms we consider. To cover algorithms that have a decreasing
learning parameter, we first show that any positive proper scoring
rule can
be interpreted as having a learning parameter $\eta$.

\begin{proposition}\label{prop:rewrite}
Let $f$ be any strictly proper scoring rule. We can write $f$ as
$f(p,r) = a +  b f'(p,r)$ with $a\in\R$, $b\in \R^+$ and $f'$ a
strictly proper scoring rule with $\min(f'(0,1), f'(1,0))=0$ and
$\max(f'(0,0), f'(1,1))=1$.  
\end{proposition}
\begin{proof}
Let $f_{min} = \min(f(0,1), f(1,0))$ and $f_{max} = \max(f(0,0),
f(1,1))=1$. Then define $f'(p,r) = \frac{f(p,r) -
  f_{min}}{f_{max}-f_{min}}$, $a=f_{min}$ and $b=f_{max}-f_{min}$. This is a positive affine translation, hence $f'$ is a strictly proper scoring rule.
\end{proof}

We call $f' : [0,1]\times\{0,1\} \rightarrow [0,1]$ a \emph{normalized} scoring rule. Using normalized scoring rules, we can define a family of scoring rules with different learning rates $\eta$.

\begin{definition}\label{def:family}
Let $f$ be any normalized strictly proper scoring rule. 
Define $\F$ as the following
family of proper scoring rules generated by $f$: 
$$\F = \{f'(p,r) = a\left(1+\eta (f(p,r) - 1)\right)  : a>0 \text{ and }
\eta\in(0,1)\}$$ 
\end{definition}

By Proposition~\ref{prop:rewrite} the union of families generated by
normalized strictly proper scoring rules cover all strictly proper
scoring rules. Using this we can now formulate the class of deterministic algorithms that are incentive compatible.

\begin{definition}[Deterministic Incentive-Compatible Algorithms]
  Let $\Ad$ be the set of deterministic algorithms that update weights
  by $w_i^{(t+1)} = a(1+\eta (f(\pt,\rt)-1))\wt$, for a normalized
  strictly proper scoring rule $f$ and $\eta\in(0,\tfrac12)$ with
  $\eta$ possibly decreasing over time. For
  $q=\sum_{i=1}^n \wt\pt/\sum_{i=1}^n \wt$, $A$ picks $\qt=0$ if
  $q<\tfrac12$, $\qt=1$ if $q>\tfrac12$ and uses any deterministic tie
  breaking rule for $q=\tfrac12$.
\end{definition}

Using this definition we can now state our main result:

\begin{theorem}\label{thm:main-lb}
For the absolute loss function,
there does not exists a deterministic and incentive-compatible algorithm
  $A\in \Ad$ with no $2$-regret.
\end{theorem}

To prove Theorem~\ref{thm:main-lb} we proceed in two steps. First
we consider strictly proper scoring rules that are symmetric with respect
to the outcomes, because they lead to a lower bound that can be
naturally interpreted by looking at the geometry of the scoring
rule. We then extend these results to cover weight-update rules that
use any (potentially asymmetric) strictly proper scoring rule. 

\subsection{Symmetric Strictly Proper Scoring Rules}
We first focus on symmetric scoring rules
in the  sense that $f(p,0)=f(1-p,1)$ for all $p\in[0,1]$. We can thus write these as $f(p)
= f(p,1) = f(1-p,0)$. Many common scoring rules are symmetric,
including the Brier rule \citep{B50}, the family of pseudo-spherical
rules (e.g. \citep{GR07}), the power family (e.g. \citep{JNW08}), and
the beta family \citep{B05} when $\alpha = \beta$. We start by
defining the scoring rule gap for normalized scoring rules, which will
determine the lower bound constant.

\begin{definition}[Scoring Rule Gap]
The {\em scoring rule gap} $\gamma$ of family $\F$ with generator $f$ is $\gamma = f(\tfrac12) - \tfrac12(f(0) + f(1))=f(\tfrac12)-\tfrac12$.
\end{definition}

The following proposition shows that for all strictly proper scoring
rules, the scoring rule gap must be strictly positive.

\begin{proposition}
	The scoring rule gap $\gamma$ of a family generated by a symmetric strictly proper scoring rule $f$ is strictly positive.
\end{proposition}
\begin{proof}
	Since $f$ is symmetric and a strictly proper scoring rule, we must have that $\tfrac12f(\tfrac12) + \tfrac12f(\tfrac12) > \tfrac12f(0) + \tfrac12f(1)$ (since an expert with belief $\tfrac12$ must have a strict incentive to report $\tfrac12$ instead of $1$). The statement follows from rewriting.
\end{proof}

We are now ready to prove our lower bound for all symmetric strictly
proper scoring rules. The interesting case is where the learning rate
$\eta\rightarrow 0$, as otherwise it is easy to prove a lower bound
bounded away from 2.

The following lemma establishes that the gap parameter is important in
proving lower bounds for IC online prediction algorithms. Intuitively, the lower bound instance exploits that experts who report $\tfrac12$ will have a higher weight (due to the scoring rule gap) than an expert who is alternatingly right and wrong with extreme reports. This means that even though the indifferent expert has the same absolute loss, she will have a higher weight and this can lead the algorithm astray. The scoring rule gap is also relevant for 
the discussion in Appendix~\ref{s:selection}.  We give partial proof of the lemma
below; the full proof appears in Appendix~\ref{s:proofs}.
\begin{lemma}
\label{lem:det-sym-lb}
Let $\F$ be a family of scoring rules generated by a symmetric
strictly proper scoring rule $f$, and let $\gamma$ be the scoring rule
gap of $\F$. In the worst case, MW with any scoring rule $f'\in\F$ with $\eta\in(0,\tfrac12)$, algorithm loss $\MT$ and expert loss $\mT$, satisfies 
	$$\MT \geq\left(2 + \frac{1}{\lceil\gamma^{-1}\rceil}\right)\cdot \mT.$$
\end{lemma}
\begin{proof}[Proof Sketch]
Let $a$, $\eta$ be the parameters of $f'$ in the family $\F$, as in
Definition~\ref{def:family}. 
Fix $T$ sufficiently large and an integer multiple of
$2\lceil\gamma^{-1}\rceil + 1$, and let $e_1$, $e_2$, and $e_3$ be
three experts. For $t=1,...,\alpha\cdot T$ where $\alpha
=\tfrac{2\lceil\gamma^{-1}\rceil}{2\lceil\gamma^{-1}\rceil+1}$ such
that $\alpha T$ is an even integer, let $p_1^{(t)} = \frac12$,
$p_2^{(t)} = 0 $, and $p_3^{(t)} = 1$. Fix any tie-breaking rule for
the algorithm. Realization $\rt$ is always the opposite of what the
algorithm chooses. 
	
Let $\Mt$ be the loss of the algorithm up to time $t$, and let $\mt$
be the loss of expert $i$. We first show that at $t'=\alpha T$,
$m_1^{(t')}=m_2^{(t')}=m_3^{(t')} = \tfrac{\alpha T}2$ and $M^{(t')} =
\alpha T$. The latter part is obvious as $\rt$ is the opposite of what
the algorithm chooses. That $m_1^{(t')} = \tfrac{\alpha T}{2}$ is
also obvious as it adds a loss of $\frac12$ at each time step. To show
that $m_2^{(t')} = m_3^{(t')} = \tfrac{\alpha T}{2}$ we do induction
on the number of time steps, in steps of two. The induction hypothesis
is that after an even number of time steps, $m_2^{(t)}=m_3^{(t)}$ and
that $w_2^{(t)}=w_3^{(t)}$. Initially, all weights are $1$ and both
experts have loss of $0$, so the base case holds. Consider the
algorithm after an even number $t$ time steps. Since
$w_2^{(t)}=w_3^{(t)}$, $p_3^{(t)} = 1-p_2^{(t)}$, and $p_1^{(t)} =
\frac12$ we have that $\sum_{i=1}^3 \wt\pt = \sum_{i=1}^3 \wt(1-\pt)$
and hence the algorithm will use its tie-breaking rule for its next
decision. Thus, either $e_2$ or $e_3$ is wrong. Wlog let's say that
$e_2$ was wrong (the other case is symmetric), so $m_2^{(t+1)} =
1+m_3^{(t+1)}$. Now at time $t+1$, $w_2^{(t+1)} =
(1-\eta)w_3^{(t+1)}<w_3^{(t+1)}$. Since $e_1$ does not express a preference, and
$e_3$ has a higher weight than $e_2$, the algorithm will follow
$e_3$'s advice. Since the realization $r^{(t+1)}$ is the opposite of
the algorithms choice, this means that now $e_3$ incurs a loss of
one. Thus $m_2^{(t+2)}= m_2^{(t+1)}$ and $w_2^{(t+2)} = w_2^{(t+1)}$
and  $m_3^{(t+2)} = 1+m_3^{(t+1)} = m_2^{(t+2)}$. The weight of expert
$e_2$ is $w_2^{(t+2)} = aa(1-\eta)w_2^{(t)}$ and the weight of expert
$e_3$ is $w_3^{(t+2)} = a(1-\eta)aw_3^{(t)}$. By the induction
hypothesis $w_2^{(t)}=w_3^{(t)}$, hence $w_2^{(t+2)}=w_3^{(t+2)}$, and since we already showed that $m_2^{(t+2)} = m_3^{(t+2)}$, this
completes the induction. 
	
Now, for $t=\alpha T + 1, ..., T$, we let $p_1^{(t)} = 1$, $p_2^{(t)}
= 0$, $p_3^{(t)} = \frac12$ and $\rt = 0$. So henceforth $e_3$ does
not provide information, $e_1$ is always wrong, and $e_2$ is always
right. If we can show that the algorithm will always follow $e_1$,
then the best expert is $e_2$ with a loss of $m_2^{(T)} =
\tfrac{\alpha T}2$, while the algorithm has a loss of $\MT = T$. If
this holds for $\alpha<1$ this proves the claim. So what's left to
prove is that the algorithm will always follow $e_1$. Note that since
$p_3^{(t)}=\frac12$ it contributes equal amounts to $\sum_{i=1}^3 \wt
\pt$ and $\sum_{i=1}^3 \wt (1-\pt)$ and is therefore ignored by the
algorithm in making its decision. So it suffices to look at $e_1$ and
$e_2$. The algorithm will pick $1$ iff $\sum_{i=1}^3 \wt (1-\pt) \leq
\sum_{i=1}^3 \wt \pt$, which after simplifying becomes $w_1^{(t)} >
w_2^{(t)}$. 
	
At time step $t$, $w_1^{(t)} = \left(a(1+\eta(f(\tfrac12) - 1))\right)^{\alpha T}(a\cdot(1-\eta))^{t-\alpha T}$ and
$w_2^{(t)} = \left(a(1-\eta)\right)^{\tfrac{\alpha
    T}2}a^{\tfrac{\alpha T}2 + t-\alpha T}$.  
	
We have that $w_1^{(t)}$ is decreasing faster in $t$ than
$w_2^{(t)}$. So if we can show that $w_1^{(T)} \geq w_2^{(T)}$ for
some $\alpha<1$, then $e_2$ will incur a total loss of $\alpha T/2$
while the algorithm incurs a loss of $T$ and the statement is
proved. This is shown in the appendix. 
\end{proof}

As a consequence of Lemma~\ref{lem:det-sym-lb}, we can calculate lower bounds for specific strictly proper scoring rules. For example, the spherical rule used in Section~\ref{ss:det-ub} is a symmetric strictly proper scoring rule with a gap parameter $\gamma=\tfrac{\sqrt 2}{2} - \tfrac{1}{2}$, and hence $1/\lceil
\gamma^{-1}\rceil=\tfrac15$. 

\begin{corollary}\label{cor:spherical}
	In the worst case, the deterministic algorithm based on the spherical rule in Section~\ref{ss:det-ub} has
	$$\MT \geq \left(2 + \tfrac15\right)\mT.$$
\end{corollary}

We revisit the scoring rule gap parameter again in Appendix~\ref{s:selection} when we discuss considerations for selecting different scoring rules.

\subsection{Beyond Symmetric Strictly Proper Scoring Rules}
We now extend the lower bound example to cover arbitrary strictly proper scoring rules. As in the previous subsection, we consider properties of normalized scoring rules to provide lower bounds that are independent of learning rate, but the properties in this subsection have a less natural interpretation.

For arbitrary strictly proper scoring rule $f'$, let $f$ be the corresponding normalized scoring rule, with parameters $a$ and $\eta$. Since $f$ is normalized, $\max\{f(0,0), f(1,1)\} = 1$ and $\min\{f(0,1), f(1,0)\}=0$. We consider 2 cases, one in which $f(0,0) = f(1,1) = 1$ and $f(0,1) = f(1,0)=0$ which is locally symmetric, and the case where at least one of those equalities does not hold.

\paragraph{\bf The semi-symmetric case.} If it is the case that $f$ has $f(0,0) = f(1,1) = 1$ and $f(0,1) = f(1,0)=0$, then $f$ has enough symmetry to prove a variant of the lower bound instance discussed just before. Define the semi-symmetric scoring rule gap as follows.

\begin{definition}[Semi-symmetric Scoring Rule Gap]
	The {\em `semi-symmetric' scoring rule gap} $\mu$ of family $\F$ with normalized generator $f$ is $\mu = \tfrac12 \left(f(\frac12,0) + f(\frac12,1)\right) - \tfrac12$.
\end{definition}

Like the symmetric scoring rule gap, $\mu>0$ by definition, as there needs to be a strict incentive to report $\tfrac12$ for experts with $\bt=\tfrac12$. Next,  observe that since $f(\tfrac12, 0), f(\tfrac12,1)\in [0,1]$ and $f(\frac12,0) + f(\frac12,1) = 1+2\mu$, it must be that $f(\frac12,0)\cdot f(\frac12,1)\geq 2\mu$. Using this it follows that:

\begin{align}
&\left(1+\eta(f(\tfrac12,0)-1)\right)\left(1+\eta(f(\tfrac12,1)-1)\right) \nonumber\\
&= 1+\eta\cdot\left(f(\tfrac12,0) + f(\tfrac12,1) - 2\right) + \eta^2\left(f(\tfrac12,0)\cdot f(\tfrac12,1) - f(\tfrac12,0) - f(\tfrac12,1) + 1\right)\nonumber\\
&= 1+\eta\cdot\left(1+2\mu - 2\right) + \eta^2\left(f(\tfrac12,0)\cdot f(\tfrac12,1)-2\mu\right)\nonumber\\
&\geq 1-\eta(1-2\mu)+ \eta^2\left(2\mu- 2\mu \right)\nonumber\\
&= 1-\eta+2\mu\eta \label{e:semisym}
\end{align}
Now this can be used in the same way as we proved the setting before:
\begin{lemma}
	\label{lem:det-asym1-lb}
	Let $\F$ be a family of scoring rules generated by a normalized
	strictly proper scoring rule $f$, with $f(0,0)=f(1,1)$ and $f(0,1)=f(1,0)$. In the worst case, MW with any scoring rule $f'$ from
	$\F$ with $\eta\in(0,\tfrac12)$ can do no better than 
	
	$$\MT \geq\left(2 + \frac{1}{\lceil\mu^{-1}\rceil}\right)\cdot \mT.$$
\end{lemma}
\begin{proof}[Proof Sketch]
	Take the same instance as used in Lemma~\ref{lem:det-sym-lb}, with $\alpha=\tfrac{2\lceil\mu^{-1}\rceil}{2\lceil\mu^{-1}\rceil+1}$. The progression of the algorithm up to $t=\alpha T$ is identical in this case, as expert $e_1$ is indifferent between outcomes, and $f(0,0)=f(1,1)$ and $f(0,1)=f(1,0)$ for experts $e_2$ and $e_3$. What remains to be shown is that the weight of $e_1$ will be higher at time $T$. At time $T$ the weights of $e_1$ and $e_2$ are:
	\begin{align*}
	a^{-T}w_1^{(T)} &= \left(1+\eta(f(\tfrac12,0) - 1)\right)^{\tfrac{\alpha T}2}\left(1+\eta(f(\tfrac12,1) - 1)\right)^{\tfrac{\alpha T}2}(1-\eta)^{(1-\alpha) T}\\
	a^{-T}w_2^{(T)} &= \left(1-\eta\right)^{\tfrac{\alpha T}2}. 
	\end{align*}
	
	Similarly to the symmetric case, wee know that $w_1^{(T)}>w_2^{(T)}$ if we can show that $$\left(1+\eta(f(\tfrac12,0) - 1)\right)^{\lceil\mu^{-1}\rceil}\left(1+\eta(f(\tfrac12,1) - 1)\right)^{\lceil\mu^{-1}\rceil}(1-\eta)>\left(1-\eta\right)^{\lceil\mu^{-1}\rceil}.$$
	By~\eqref{e:semisym}, it suffices to show that $\left(1-\eta+2\mu\eta\right)^{\lceil\mu^{-1}\rceil}(1-\eta)>\left(1-\eta\right)^{\lceil\mu^{-1}\rceil}$, which holds by following the derivation in the proof of Lemma~\ref{lem:det-sym-lb} given in the appendix, starting at \eqref{e:sym}.
\end{proof}

\paragraph{\bf The asymmetric case.} We finally consider the setting where the weight-update rule is not symmetric, nor is it symmetric evaluated only at the extreme reports. The lower bound that we show is based on the amount of asymmetry at these extreme points, and is parametrized as follows.

\begin{definition}\label{def:asym}
	Let $c>d$ be parameters of a normalized strictly proper scoring rule $f$, such that $c=1-\max\{f(0,1), f(1,0)\}$ and $d=1-\min\{f(0,0),f(1,1)\}$.
\end{definition}

Scoring rules that are not covered by Lemmas~\ref{lem:det-sym-lb}~or~\ref{lem:det-asym1-lb} must have either $c<1$ or $d>0$ or both.
The intuition behind the lower bound instance is that two experts who have opposite predictions, and are alternatingly right and wrong, will end up with different weights, even though they have the same loss. We use this to show that eventually one expert will have a lower loss, but also a lower weight, so the algorithm will follow the other expert. This process can be repeated to get the bounds in the Lemma below. The proof of the lemma appears in the appendix.
\begin{lemma}
	\label{lem:det-asym2-lb}
	Let $\F$ be a family of scoring rules generated by a normalized
	strictly proper scoring rule $f$, with not both $f(0,0)=f(1,1)$ and $f(0,1)=f(1,0)$ and parameters $c$ and $d$ as in Definition~\ref{def:asym}. In the worst case, MW with any scoring rule $f'$ from
	$\F$ with $\eta\in(0,\tfrac12)$ can do no better than 
	$$\MT \geq\left(2 + \max\{\tfrac{1-c}{2c}, \tfrac{d}{4(1-d)}\}\right)\cdot \mT.$$
\end{lemma}

Theorem~\ref{thm:main-lb} now follows from combining the previous three lemmas.

\begin{proof}[Proof of Theorem~\ref{thm:main-lb}]
	Follows from combining Lemmas~\ref{lem:det-sym-lb},~\ref{lem:det-asym1-lb}~and~\ref{lem:det-asym2-lb}.
\end{proof}


\section{The Cost of Selfish Experts for Non-IC Algorithms}
\label{s:det-non-ic-lb}
What about non-incentive-compatible algorithms?  Could it be that,
even with experts reporting
strategically instead of honestly, 
there is a deterministic no $2$-regret algorithm (or a randomized algorithm with
vanishing regret), to match the classical results for honest experts?
Proposition~\ref{prop:standard-det-lb} shows that the standard
algorithm fails to achieve such a regret bound, 
but maybe some other non-IC algorithm does?

Typically, one would show that this is not the case by a ``revelation principle'' argument: if there exists some (non-IC) algorithm $A$ with good guarantees, then we can construct an algorithm $B$ which takes private values as input, and runs algorithm $A$ on whatever reports a self-interested agent would have provided to $A$. It does the strategic thinking for agents, and hence $B$ is an IC algorithm with the same performance as $A$. This means that generally, whatever performance is possible with non-IC algorithms can be achieved by IC algorithms as well, thus lower bounds for IC algorithms translate to lower bounds for non-IC algorithms. In our case however, the reports impact both the weights of experts as well as the decision of the algorithm simultaneously. Even if we insist on keeping the weights in $A$ and $B$ the same, the decisions of the algorithms may still be different. Therefore, rather than relying on a simulation argument, we give a direct proof that, under mild technical conditions, non-IC deterministic algorithms cannot be no $2$-regret.\footnote{Similarly to Price of Anarchy (PoA) bounds, e.g. \citep{RT07}, the results here show the harm of selfish behavior. Unlike PoA bounds, we sidestep the question of equilibrium concepts and our results are additive rather than multiplicative.} As in the previous section, we focus on deterministic algorithms;
Section~\ref{s:rand} translates
these lower bounds to randomized algorithms,
where they imply that no vanishing-regret
algorithms exist.

The following definition captures
how players are incentivized to report differently from their beliefs.

\begin{definition}[Rationality Function]
	For a weight update function $f$, let $\rho_f : [0,1] \rightarrow [0,1]$ be the function from beliefs to predictions, such that reporting $\rho_f(b)$ is rational for an expert with belief $b$.
\end{definition}

We restrict our attention here on rationality functions that are proper functions, meaning that each  belief leads to a single prediction. Note that for incentive compatible weight update functions, the rationality function is simply the identity function.

Under mild technical conditions on the rationality function, we show our main lower bound for (potentially non-IC) algorithms.\footnote{This holds even when the learning rate is
	parameterized similarly to Definition~\ref{def:family}, as the
	rationality function does not change for different learning rates
	due to the linearity of the expectation operator.}

\begin{theorem}
	\label{thm:non-ic-main-lb}
	For a weight update function $f$ with continuous or non-strictly increasing rationality function $\rho_f$, there is no deterministic no $2$-regret algorithm.
\end{theorem}
Note that Theorem~\ref{thm:non-ic-main-lb} covers the standard algorithm, as well as other common update rules such as the Hedge update rule $f_\text{Hedge}(\pt,\rt) = e^{-\eta|\pt-\rt|}$ \citep{FS97}, and all IC methods, since they have the identity rationality function (though the bounds in Thm~\ref{thm:main-lb} are stronger).

We start with a
proof that any algorithm with non-strictly increasing rationality
function must have worst-case loss strictly more than twice the best expert in hindsight. Conceptually, the proof is a generalization of the proof for
Proposition~\ref{prop:standard-det-lb}.  

\begin{lemma}
	\label{lem:det-non-ic-lb-non-monotone}
Let $f$ be a weight update function with a non-strictly increasing
rationality function $\rho_f$, such that there exists $b_1 < b_2$ with
$\rho_f(b_1) \geq \rho_f(b_2)$. 
For every deterministic algorithm, in the worst case $$\MT \geq (2 + |b_2-b_1|)\mT.$$
\end{lemma}
\begin{proof}
	Fix, $f$, $b_1$ and $b_2$ such that $\rho_f(b_1) \geq \rho_f(b_2)$ with $b_1 < b_2$. Let $\pi_1 = \rho_f(b_1)$, $\pi_2 = \rho_f(b_2)$, $b_0 = 1-\frac{b_2 + b_1}{2}$, and $\pi_0 = \rho_f(b_0)$.
	
	Let there be two experts $e_0$ and $e_1$. Expert $e_0$ always predicts $\pi_0$ with belief $b_0$. If $\pi_1=\pi_2$, $e_1$ predicts $\pi_1$ (similar to Proposition~\ref{prop:standard-det-lb}, we first fix the predictions of $e_1$, and will give consistent beliefs later). Otherwise $\pi_1 > \pi_2$, and expert $e_1$ has the following beliefs (and corresponding predictions) at time $t$:
	$$
		b_1^{(t)}=\begin{cases}
		b_1 & \text{if }\frac{w_0^{(t)}\pi_0 + w_1^{(t)}\pi_2}{w_0^{(t)} + w_1^{(t)}}\geq\frac12\\
			b_2 & \text{otherwise}
		\end{cases}
	$$
	The realizations are opposite of the algorithm's decisions.
	
	We now fix the beliefs of $e_1$ in the case that $\pi_1=\pi_2$. Whenever $\rt=1$, set expert $e_1$'s belief to $b_2$, and whenever $\rt=0$, set her belief to $b_1$. Note that the beliefs indeed lead to the predictions she made, by the fact that $\pi_1 = \rho_f(b_1) = \rho_f(b_2)$.
	
	For the case where $\pi_1 > \pi_2$, if $(w_0^{(t)}\pi_0 + w_1^{(t)}\pi_2)/(w_0^{(t)} + w_1^{(t)})\geq\tfrac12$ then $e_1$'s belief will be $b_1$ leading to a report of $\pi_1$ and as $\pi_1>\pi_2$ it must hold that $(w_0^{(t)}\pi_0 + w_1^{(t)}\pi_1)(w_0^{(t)} + w_1^{(t)})>\tfrac12$, hence the algorithm will certainly choose $1$, so the realization is $0$. Conversely, if $(w_0^{(t)}\pi_0 + w_1^{(t)}\pi_2)(w_0^{(t)} + w_1^{(t)})<\frac12$, then the belief of $e_1$ will be $b_2$ and her report will lead the algorithm to certainly choose $0$, so the realization is $1$. So in all cases, if the realization is $1$, then the belief of expert $e_1$ is $b_2$ and otherwise it is $b_1$.
	
	The total number of mistakes $\MT$ for the algorithm after $T$ time steps is $T$ by definition. Every time the realization was $1$, $e_0$ will incur loss of $\tfrac{b_1+b_2}2$ and $e_1$ incurs a loss of $1-b_2$, for a total loss of $1-b_2 + \tfrac{b_1+b_2}2 = 1 - \tfrac{b_2 - b_1}{2}$. Whenever the realization was $0$, $e_0$ incurs a loss of $1-\tfrac{b_1+b_2}2$ and $e_1$ incurs a loss of $b_1$ for a total loss of $1-\tfrac{b_1+b_2}2 + b_1 = 1 - \tfrac{b_2 - b_1}{2}$.
	
	So the total loss for \emph{both} of the experts is $\left(1-\tfrac{b_2-b_1}{2}\right)\cdot T$, so the best expert in hindsight has $\mT \leq \frac{1}{2}\left(1-\tfrac{b_2-b_1}{2}\right)\cdot T$. Rewriting yields the claim.
\end{proof}

For continuous rationality functions, we can generalize the results in
Section~\ref{s:det-lb} using a type of simulation argument.
First, we address some edge cases.

\begin{proposition}
	\label{prop:edge-cases}
	For a weight update function $f$ with continuous strictly increasing rationality function $\rho_f$,
	\begin{enumerate}
		\item 
the regret is unbounded unless $\rho_f(0)<\tfrac12<\rho(1)$; and 
		\item if $\rho_f(b)=\tfrac12$ for $b\neq\tfrac12$, the worst-case loss of the algorithm satisfies $\MT \geq \left(2+|b-1/2|\right)\mT.$
	\end{enumerate} 
\end{proposition}
\begin{proof}
	First, assume that it does not hold that $\rho_f(0)<\tfrac12<\rho_f(1)$. Since $\rho_f(0)<\rho_f(1)$ by virtue of $\rho_f$ being strictly increasing, it must be that either $\tfrac12\leq\rho_f(0) < \rho_f(1)$ or $\rho_f(0)<\rho_f(1)\leq\tfrac12$. Take two experts with $b_1^{(t)}=0$ and $b_2^{(t)}=1$. Realizations are opposite of the algorithm's predictions. Even though the experts have opposite beliefs, their predictions agree (potentially with one being indifferent), so the algorithm will consistently pick the same prediction, whereas one of the two experts will never make a mistake. Therefore the regret is unbounded.
	
	As for the second statement. Since $\rho_f(0)<\tfrac12<\rho_f(1)$, there is some $b$ such that $\rho_f(b)=\tfrac12$. Wlog, assume $b<\tfrac12$ (the other case is analogous). Since $\rho_f$ is continuous and strictly increasing, $\rho_f(\tfrac{b+1/2}2)>\tfrac12$ while $\tfrac{b+1/2}2<\tfrac12$. Take one expert $e_1$ with belief $b^{(t)}=\tfrac{b+1/2}2<\tfrac12$, who will predict $p^{(t)}=\rho_f(\tfrac{b+1/2}2)>\tfrac12$. Realizations are opposite of the algorithms decisions, and the algorithms decision is consistently 1, due to there only being one expert, and that expert putting more weight on 1. However, the absolute loss of the expert is only $\tfrac12 - \tfrac{|b-1/2}2$ at each time step. Summing over the timesteps and rewriting yields the claim.
\end{proof}

We are now ready to prove the main result in this section. The proof gives lower bound constants that are similar (though not identical) to the constants given in Lemmas~\ref{lem:det-sym-lb}, \ref{lem:det-asym1-lb} and \ref{lem:det-asym2-lb}, though due to a reparameterization the factors are not immediately comparable.
The proof appears in the appendix.

\begin{theorem}
	\label{thm:det-non-ic-lb-continuous}
	For a weight update function $f$ with continuous strictly
        increasing rationality function $\rho_f$, with
        $\rho_f(0)<\tfrac12<\rho_f(1)$ and $\rho_f(\tfrac12) =
        \tfrac12$, there is no deterministic no $2$-regret algorithm.
\end{theorem}

Theorem~\ref{thm:non-ic-main-lb} now follows from Lemma~\ref{lem:det-non-ic-lb-non-monotone}, Proposition~\ref{prop:edge-cases} and Theorem~\ref{thm:det-non-ic-lb-continuous}.


\section{Randomized Algorithms: Upper and Lower Bounds}
\label{s:rand}

\subsection{Impossibility of Vanishing Regret}

We now consider randomized online learning algorithms, which can
typically achieve better worst-case guarantees than deterministic
algoritms.
For example, with honest experts, there are randomized algorithms with
worst-case loss
$\MT \leq \left(1+O\left(\tfrac{1}{\sqrt T}\right)\right)\mT$, which
means that the regret with respect to the best expert in hindsight is
vanishing as $T \rightarrow \infty$.  Unfortunately, the lower bounds in
Sections~\ref{s:det-lb}~and~\ref{s:det-non-ic-lb} below imply that no such
result is possible for randomized algorithms.

\begin{corollary}
	\label{lem:ic-ran-lb}
	Any incentive compatible randomized weight-update algorithm or
        non-IC randomized algorithm with continuous or non-strictly
        increasing rationality function cannot be no $1$-regret.
\end{corollary}
\begin{proof}
	We can use the same instances as for 
Theorems~\ref{thm:main-lb} and~\ref{thm:det-non-ic-lb-continuous} and
Lemma~\ref{lem:det-non-ic-lb-non-monotone}
(whenever the algorithm was indifferent, the realizations were defined using the algorithm's tie-breaker rule; in the current setting simply pick any realization, say $r^t=1$).
	
Whenever the algorithm made a mistake, it was because
$\sum_i w_i^ts_i^t \geq \frac{1}{2}\sum_i w_i^t$. Therefore, in the
randomized setting, it will still incur an expected loss of at least
$\frac12$. Therefore the total expected loss of the randomized
algorithm is at least
half that of the deterministic algorithm. Since
the approximation factor for the latter is bounded away from $2$ in
all cases in
Theorems~\ref{thm:main-lb} and~\ref{thm:det-non-ic-lb-continuous} and
Lemma~\ref{lem:det-non-ic-lb-non-monotone}, 
in these cases the worst-case loss of a randomized
algorithm satisfies $\MT \geq (1+\Omega(1))\mT.$
\end{proof}

\subsection{An Incentive-Compatible Randomized Algorithm for Selfish Experts}
While we cannot hope to achieve a no-regret algorithm for online prediction with selfish experts, we can do better than the deterministic algorithm from Section~\ref{s:det-ub}. We now focus on the more general class of algorithms where the
algorithm can make any prediction $\qt\in[0,1]$ and incurs a loss of
$|\qt-\rt|$. 
We give a randomized algorithm based on the Brier
strictly proper scoring rule with loss at most $2.62$ times that
of the best expert as $T\rightarrow \infty$. 

Perhaps the most natural choice for a randomized algorithm is to simply report a prediction of
$\qt = \sum_{i=1}^n \wt\pt / \sum_{j=1}^n w_j^{(t)}$. However, this is
problematic when the experts are highly confident and correct in their
predictions. By the definition of a (bounded) strictly proper scoring
rule, $\frac{d}{d\pt} f(\pt,1)$ is $0$ at $1$ (and similarly the
derivative is $0$ around $0$ for a realization of $0$). This means
that experts that are almost certain and correct will not have their
weight reduced meaningfully, and so the proof that uses the potential
function does not go through. 

This motivates looking for an algorithm where the sum of weights of
experts is guaranteed to decrease significantly whenever the algorithm
incurs a loss. Consider the following generalization of RWM that
rounds predictions to the nearest integer if they are with $\theta$ of
that integer. 

\begin{definition}[$\theta$-randomized weighted majority]
	\label{def:A}
	Let $\Ar$ be the class of algorithms that maintains expert weights as in Definition~\ref{def:wuopa}. Let $b^{(t)} = \sum_{i=1}^n \frac{\wt}{\sum_{j=1}^n w_j^{(t)}}\cdot \pt$ be the weighted predictions. For parameter $\theta\in [0,\tfrac12]$ the algorithm chooses $1$ with probability
	$$\Mpt = \begin{cases}
	0 & \text{if }\Mbt \leq \theta\\
	\Mbt & \text{if }\theta < \Mbt \leq 1-\theta\\ 
	1 & \text{otherwise.}
	\end{cases}$$
\end{definition}
We call algorithms in $\Ar$ $\theta$-RWM algorithms. We'll use a
$\theta$-RWM algorithm with the Brier rule. Recall that $\st =
|\pt-\rt|$; the Brier rule is defined as: 

\begin{align}\label{eq:brier}
f_\text{Br}(\pt,\rt) = 1-\eta\left(\frac{(\pt)^2 + (1-\pt)^2 +1}{2}-(1-\st) \right).
\end{align}

\begin{theorem}\label{thm:rand-ub}
Let $A\in\Ar$ be a $\theta$-RWM algorithm with the Brier weight update
rule $f_\text{Br}$ and $\theta=0.382$ and with
$\eta=O(1/\sqrt T)\in(0,\tfrac12)$. $A$ has no $2.62$-regret.
\end{theorem}
The proof appears in the appendix.


\section{Simulations}
\label{s:simulations}



The theoretical results presented so far indicate that when faced with selfish experts, one should use an IC weight update rule, and ones with smaller scoring rule gap are better. Two objections to these conclusions are: first, the presented results are \emph{worst-case}, and different instances are used to obtain the bounds for different scoring rules. A priori it is not obvious that for an arbitrary (non worst-case) input, the regret of different scoring rules follow the same relative ordering. It is of particular interest to see if the non-IC standard weight-update rule does better or worse than the IC methods proposed in this paper. Second, there is a gap between our upper and lower bounds for IC rules. It is therefore informative to look at different instances for which we expect our algorithms to do badly, to see if the performance is closer to the upper bound or to the lower bound.

\subsection{Data-Generating Processes}
To address these two concerns, we look at three different data-generating processes.

\paragraph{Hidden Markov Model.} The experts are represented by a simple two-state hidden Markov model (HMM) with a ``good'' state and a ``bad'' state. We first flip a fair coin to determine the realization $\rt$. For $\rt= 0$ (otherwise beliefs are reversed), in the good state expert $i$ believes $\bt \sim \min\{\text{Exp}(1)/5, 1\}$: the belief is exponentially distributed with parameter $\lambda=1$, values are rescaled by $\tfrac15$ and clamped between $0$ and $1$. In the bad state, expert $i$ believes $\bt\sim \text{U}[\tfrac12,1]$.  The transition probabilities to move to the other state are $\tfrac1{10}$ for both states. This data generating process models that experts that have information about the event are more accurate than experts who lack the information.

\paragraph{Lower Bound Instance.} The lower bound instance described in the proof of Lemma~\ref{lem:det-sym-lb}.

\paragraph{Greedy Lower Bound.} A greedy version of the lower bound described the proof of Lemma~\ref{lem:det-sym-lb}. There are 3 experts, one ($e_0$) who is mostly uninformative, and two ($e_1$ and $e_2$) who are alternating correct and incorrect. Whenever the weight of $e_0$ is ``sufficiently'' higher than that of $e_1$ and $e_2$, we have ``punish the algorithm'' by making $e_0$ wrong twice: $b_0^{(t)} = 0$, $b_1^{(t)} = 1$, $b_2^{(t)} = \tfrac12$, $\rt = 1$, and $b_0^{(t+1)} = 0$, $b_1^{(t)} = \tfrac12$, $b_0^{(t)} = 1$, $\rt = 1$. ``Sufficiently'' here means that weight of $e_0$ is high enough for the algorithm to follow its advice during both steps.

\subsection{Results}
\paragraph{Hidden Markov Model Data.} In Figure~\ref{f:hmm} we show the regret as a function of time for  the standard weight-update function, the Brier scoring rule, the spherical scoring rule, and a scoring rule from the Beta family \citep{B05} with $\alpha = \beta = \tfrac12$. The expert's report $\pt$ for the IC methods correspond to their belief $\bt$, whereas for the standard weight-update rule, the expert reports $\pt=1$ if $\bt\geq \tfrac{1}{2}$ and $\pt=0$ otherwise. The $y$ axis is the ratio of the total loss of each of the algorithms to the performance of the best expert at that time. For clarity, we include the performance of the best expert at each time step, which by definition is 1 everywhere. The plot is for $10$ experts, $T=10,000$, $\eta = 10^{-2}$, and the randomized\footnote{Here we use the regular RWM algorithm, so in the notation of Section~\ref{s:rand} we have $\theta=0$.} versions of the algorithms (we return to why in a moment), averaged over 30 runs.

From the plot, we can see that each of the IC methods does significantly better than the standard weight-update algorithm. Whereas the standard weight-update rule levels off between $1.15$ and $1.2$, all of the IC methods dip below a regret of $1.05$ at $T=2,000$ and hover around $1.02$ at $T=10,000$. This trend continues and at $T=200,000$ (not shown in the graph), the IC methods have a regret of about $1.003$, whereas the regret for the standard algorithm is still $1.14$. This gives credence to the notion that failing to account for incentive issues is problematic beyond the worst-case bounds presented earlier.

Moreover, the plot shows that while there is a worst-case lower bound for the IC methods that rules out no-regret, for quite natural synthetic data, the loss of all the IC algorithms approaches that of the best expert in hindsight, while the standard algorithm fails to do this. It curious to note that the performance of all IC methods are comparable (at least for this data-generating process). This seems to indicate that eliciting the truthful beliefs of the experts is more important than the exact weight-update rule. 

Finally, note that the results shown here are for randomized weighted majority, using the different weight-update rules. For the deterministic version of the algorithms the difference between the non-IC standard weight-update rules and the IC ones is even starker. Different choices for the transition probabilities of the HMM, and different distributions, e.g. the bad state has $\bt\sim \text{U}[0,1]$, give similar results to the ones presented here.

\begin{figure}
	\centering
	\includegraphics[scale=1]{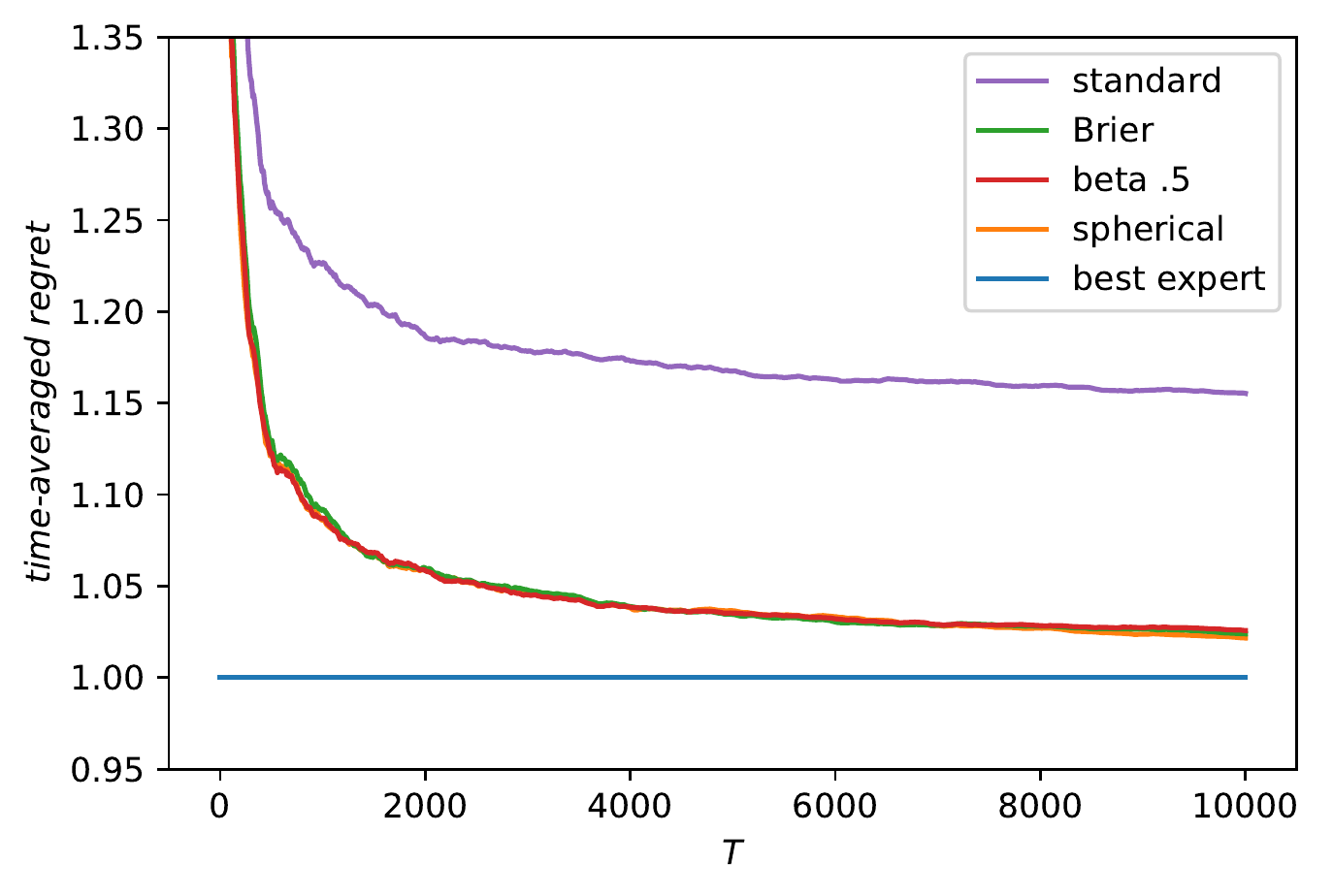}
	\caption{The time-averaged regret for the HMM data-generating process.}
	\label{f:hmm}
\end{figure}

\begin{figure}
	\centering
	\includegraphics[scale=1]{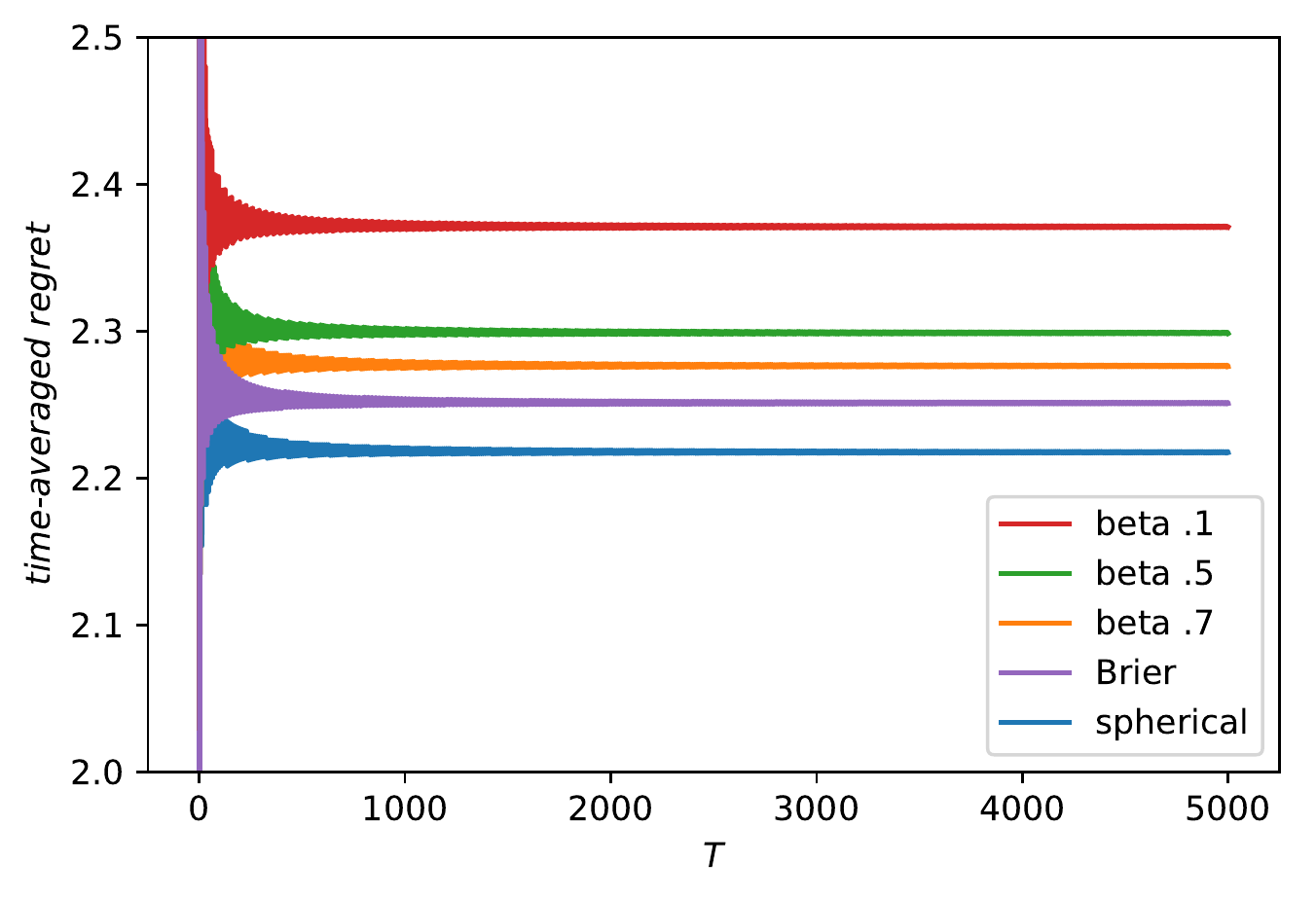}
	\caption{Regret for the greedy lower bound instance.}
	\label{f:lbs}
\end{figure}

\begin{table}
	\centering
	\caption{Comparison of lower bound results with simulation. The simulation is run for $T=10,000, \eta=10^{-4}$ and we report the average of $30$ runs. For the lower bounds, the first number is the lower bound from Lemma~\ref{lem:det-sym-lb}, i.e. $2+ \tfrac{1}{\lceil\gamma^{-1}\rceil}$, the second number (in parentheses) is $2+\gamma$.}
	\label{t:lbs}
	\footnotesize
	\begin{tabular}{|c|c|c|c|c|c|c|c|}
		\hline 
		& Beta $.1$ & Beta $.3$ & Beta $.5$ & Beta $.7$ & Beta $.9$ & Brier ($\beta=1$) & Spherical \\ 
		\hline \hline
		Greedy LB Sim & 2.3708 & 2.3283 & 2.2983 & 2.2758 & 2.2584 & 2.2507 & 2.2071 \\ 
		\hline 
		LB Simulation & $2.4414$ & $2.3657$ & $2.3186$ & $2.2847$ & $2.2599$ & $2.2502$ & 2.2070 \\ 
		\hline 
		Lem~\ref{lem:det-sym-lb} LB & $2.33\ (2.441)$ & $2.33\ (2.365)$ & $ 2.25\ (2.318)$ & $2.25\ (2.285)$ & $2.25\ (2.260)$ & $2.25$ & $2.2\ (2.207)$ \\ 
		\hline 
	\end{tabular} 
\end{table}
\paragraph{Comparison of LB Instances.} Now let's focus on the performance of different IC algorithms. First, in Figure~\ref{f:lbs} we show the regret for different algorithms on the greedy lower bound instance. Note that this instance is different from the one used in the proof of Lemma~\ref{lem:det-sym-lb}, but the regret is very close to what is obtained there. In fact, when we look at Table~\ref{t:lbs}, we can see that very closely traces $2+\gamma$. In Table~\ref{t:lbs} we can also see the numerical results for the lower bound from Lemma~\ref{lem:det-sym-lb}. In fact, for the analysis, we needed to use $\lceil\gamma^{-1}\rceil$ when determining the first phase of the instance. When we use $\gamma$ instead numerically, the regret seems to trace $2+\gamma$ quite closely, rather than the weaker proven lower bound of $2+\tfrac{1}{\lceil\gamma^{-1}\rceil}$. By using two different lower bound constructions, we can see that the analysis of Lemma~\ref{lem:det-sym-lb} is essentially tight (up to the rounding of $\gamma$), though this does not exclude the possibility that stronger lower bounds are possible using more properties of the scoring rules (rather than only the scoring rule gap $\gamma$). In these experiments (and others we have performed), the regret of IC methods never exceeds the lower bound we proved in Lemma~\ref{lem:det-sym-lb}. Closing the gap between the lower and upper bound requires finding a different lower bound instance, or a better analysis for the upper bound.

\newpage

\bibliographystyle{plainnat}
\bibliography{references}
\clearpage
\appendix
\section{Proofs}
\label{s:proofs}

\subsection{Proof of Theorem~\ref{thm:spherical-det-ub}}
Let $A$ be the WM algorithm that updates weights according to \eqref{eq:spherical} for $\eta<\frac12$. Let $M^T$ be the total loss of $A$ and $m_i^T$ the total loss of expert $i$. Then for each expert $i$, we have
$$M^T \leq (2+\sqrt 2)\left((1+\eta)m_i^T + \frac{\ln n}{\eta}\right).$$
\begin{proof}
	We use an intermediate potential function $\Phi^{(t)} = \sum_i \wt$. Whenever the algorithm incurs a loss, the potential must decrease substantially. For the algorithm incur a loss, it must have picked the wrong outcome. Therefore it loss $|\rt - t^{(t)}|=1$ and $\sum_i \wt \st \geq \frac12\cdot\Phi^{(t)}$. We use this to show that in those cases the potential drops significantly:
	\begin{align*}
	\Ptp &= \sum_i \left(1-\eta\left(1-\frac{1-\st}{\sqrt{\left(\pt\right)^2 + \left(1-\pt\right)^2}}\right)\right)\cdot \wt\\
	&\leq \sum_i \left(1-\eta\left(1 -     \sqrt 2\left(1-\st\right)\right)\right)\cdot \wt &\hspace{-.35cm}\left(\text{since $\min_x x^2 + (1-x)^2=\frac12$}\right)\\
	&=\left(1-\eta\right)\Pt +\sqrt2\eta \sum_i \left(1-\st\right) \wt\\
	&\leq\left(1-\eta\right)\Pt +\frac{\sqrt2\eta}{2}\Pt  &\hspace{-2cm}\left(\text{since $\sum_i \wt\left(1-\st\right) \leq \frac12 \Pt$}\right)\\
	&=\left(1-\frac{2-\sqrt2}{2}\eta\right)\Pt
	\end{align*}
	Since initially $\Phi^0=n$, after $\MT$ mistakes, we have:
	\begin{align}
	\label{eq:ub}
	\Phi^T&\leq n \left(1-\frac{2-\sqrt2}{2}\eta\right)^{\MT}.
	\end{align}
	Now, let's bound the final weight of expert $i$ in terms of the number of mistakes she made:
	\begin{align*}
	\wT&= \prod_t \left(1-\eta\left(1-\frac{1-\st}{\sqrt{\left(\pt\right)^2 + \left(1-\pt\right)^2}}\right)\right)\\
	&\geq \prod_t \left(1-\eta \st\right)& \left(\text{since $\max_{x\in[0,1]} x^2 + (1-x)^2=1$}\right)\\
	&\geq \prod_t (1-\eta)^{\st} & \hspace{-2cm}\left(\text{since } 1-\eta x \geq (1-\eta)^x \text{ for }x\in[0,1]\right)\\
	&= (1-\eta)^{\sum_t \st}\\
	&= (1-\eta)^{\mT}
	\end{align*}
	Combining this with $\wt \leq \Pt$ and \eqref{eq:ub}, and taking natural logarithms of both sides we get:
	\begin{align*}
	\ln \left( (1-\eta)^{\mT}\right)&\leq \ln\left(n\left(1-\frac{2-\sqrt 2}{2}\eta\right)^{\MT}\right)\\
	\mT\cdot \ln(1-\eta) &\leq \MT\cdot \ln\left(1-\frac{2-\sqrt 2}{2}\eta\right) + \ln n\\
	\mT\cdot \left(-\eta - \eta^2\right) &\leq \MT\cdot \ln\left(\exp\left(-\frac{2-\sqrt 2}{2}\eta\right)\right) + \ln n\\
	\mT\cdot \left(-\eta - \eta^2\right) &\leq \MT\cdot -\frac{2-\sqrt 2}{2}\eta + \ln n\\
	\MT &\leq \left(\frac{2}{2-\sqrt 2}\right)\cdot\left((1+\eta)\mT + \frac{\ln n}{\eta}\right) 
	\end{align*}
	where in the third inequality we used $-\eta-\eta^2 \leq \ln(1-\eta)$ for $\eta\in(0,\tfrac12)$.
	Rewriting the last statement proves the claim.
\end{proof}

\subsection{Proof of Lemma~\ref{lem:det-sym-lb}}
Let $\F$ be a family of scoring rules generated by a symmetric
strictly proper scoring rule $f$, and let $\gamma$ be the scoring rule
gap of $\F$. In the worst case, MW with any scoring rule $f'\in\F$ with $\eta\in(0,\tfrac12)$, algorithm loss $\MT$ and expert loss $\mT$, satisfies 
$$\MT \geq\left(2 + \frac{1}{\lceil\gamma^{-1}\rceil}\right)\cdot \mT.$$
\begin{proof}
	Let $a$, $\eta$ be the parameters of $f'$ in the family $\F$, as in
	Definition~\ref{def:family}. 
	Fix $T$ sufficiently large and an integer multiple of
	$2\lceil\gamma^{-1}\rceil + 1$, and let $e_1$, $e_2$, and $e_3$ be
	three experts. For $t=1,...,\alpha\cdot T$ where $\alpha
	=\tfrac{2\lceil\gamma^{-1}\rceil}{2\lceil\gamma^{-1}\rceil+1}$ such
	that $\alpha T$ is an even integer, let $p_1^{(t)} = \frac12$,
	$p_2^{(t)} = 0 $, and $p_3^{(t)} = 1$. Fix any tie-breaking rule for
	the algorithm. Realization $\rt$ is always the opposite of what the
	algorithm chooses. 
	
	Let $\Mt$ be the loss of the algorithm up to time $t$, and let $\mt$
	be the loss of expert $i$. We first show that at $t'=\alpha T$,
	$m_1^{(t')}=m_2^{(t')}=m_3^{(t')} = \tfrac{\alpha T}2$ and $M^{(t')} =
	\alpha T$. The latter part is obvious as $\rt$ is the opposite of what
	the algorithm chooses. That $m_1^{(t')} = \tfrac{\alpha T}{2}$ is
	also obvious as it adds a loss of $\frac12$ at each time step. To show
	that $m_2^{(t')} = m_3^{(t')} = \tfrac{\alpha T}{2}$ we do induction
	on the number of time steps, in steps of two. The induction hypothesis
	is that after an even number of time steps, $m_2^{(t)}=m_3^{(t)}$ and
	that $w_2^{(t)}=w_3^{(t)}$. Initially, all weights are $1$ and both
	experts have loss of $0$, so the base case holds. Consider the
	algorithm after an even number $t$ time steps. Since
	$w_2^{(t)}=w_3^{(t)}$, $p_3^{(t)} = 1-p_2^{(t)}$, and $p_1^{(t)} =
	\frac12$ we have that $\sum_{i=1}^3 \wt\pt = \sum_{i=1}^3 \wt(1-\pt)$
	and hence the algorithm will use its tie-breaking rule for its next
	decision. Thus, either $e_2$ or $e_3$ is wrong. Wlog let's say that
	$e_2$ was wrong (the other case is symmetric), so $m_2^{(t+1)} =
	1+m_3^{(t+1)}$. Now at time $t+1$, $w_2^{(t+1)} =
	(1-\eta)w_3^{(t+1)}<w_3^{(t+1)}$. Since $e_1$ does not express a preference, and
	$e_3$ has a higher weight than $e_2$, the algorithm will follow
	$e_3$'s advice. Since the realization $r^{(t+1)}$ is the opposite of
	the algorithms choice, this means that now $e_3$ incurs a loss of
	one. Thus $m_2^{(t+2)}= m_2^{(t+1)}$ and $w_2^{(t+2)} = w_2^{(t+1)}$
	and  $m_3^{(t+2)} = 1+m_3^{(t+1)} = m_2^{(t+2)}$. The weight of expert
	$e_2$ is $w_2^{(t+2)} = aa(1-\eta)w_2^{(t)}$ and the weight of expert
	$e_3$ is $w_3^{(t+2)} = a(1-\eta)aw_3^{(t)}$. By the induction
	hypothesis $w_2^{(t)}=w_3^{(t)}$, hence $w_2^{(t+2)}=w_3^{(t+2)}$, and since we already showed that $m_2^{(t+2)} = m_3^{(t+2)}$, this
	completes the induction. 
	
	Now, for $t=\alpha T + 1, ..., T$, we let $p_1^{(t)} = 1$, $p_2^{(t)}
	= 0$, $p_3^{(t)} = \frac12$ and $\rt = 0$. So henceforth $e_3$ does
	not provide information, $e_1$ is always wrong, and $e_2$ is always
	right. If we can show that the algorithm will always follow $e_1$,
	then the best expert is $e_2$ with a loss of $m_2^{(T)} =
	\tfrac{\alpha T}2$, while the algorithm has a loss of $\MT = T$. If
	this holds for $\alpha<1$ this proves the claim. So what's left to
	prove is that the algorithm will always follow $e_1$. Note that since
	$p_3^{(t)}=\frac12$ it contributes equal amounts to $\sum_{i=1}^3 \wt
	\pt$ and $\sum_{i=1}^3 \wt (1-\pt)$ and is therefore ignored by the
	algorithm in making its decision. So it suffices to look at $e_1$ and
	$e_2$. The algorithm will pick $1$ iff $\sum_{i=1}^3 \wt (1-\pt) \leq
	\sum_{i=1}^3 \wt \pt$, which after simplifying becomes $w_1^{(t)} >
	w_2^{(t)}$. 
	
	At time step $t$, $w_1^{(t)} = \left(a(1+\eta(f(\tfrac12) - 1))\right)^{\alpha T}(a\cdot(1-\eta))^{t-\alpha T}$ and
	$w_2^{(t)} = \left(a(1-\eta)\right)^{\tfrac{\alpha
			T}2}a^{\tfrac{\alpha T}2 + t-\alpha T}$.  
	
	We have that $w_1^{(t)}$ is decreasing faster in $t$ than
	$w_2^{(t)}$. So if we can show that $w_1^{(T)} \geq w_2^{(T)}$ for
	some $\alpha<1$, then $e_2$ will incur a total loss of $\alpha T/2$
	while the algorithm incurs a loss of $T$ and the statement is
	proved.
	
	We have that $w_1^{(t)}$ is decreasing faster in $t$ than $w_2^{(t)}$. So if we can show that at time $T$, $w_1^{(T)} \geq w_2^{(T)}$ for some $\alpha<1$, then $e_2$ will incur a total loss of $\alpha T$ while the algorithm incurs a loss of $T$ and the statement is proved. First divide both weights by $a^T$ so that we have
	\begin{align*}
	a^{-T}w_1^{(T)} &= \left(1+\eta(f(\tfrac12) - 1)\right)^{\alpha T}(1-\eta)^{(1-\alpha) T}\\
	a^{-T}w_2^{(T)} &= \left(1-\eta\right)^{\tfrac{\alpha T}2}. 
	\end{align*}
	
	Let $\alpha = \frac{2\lceil\gamma^{-1}\rceil}{2\lceil\gamma^{-1}\rceil+1}$ and recall that $T = k\cdot\left(2\lceil\gamma^{-1}\rceil+1\right)$ for positive integer $k$. Thus we can write
	\begin{align*}
	a^{-T}w_1^{(T)} &= \left(1+\eta(f(\tfrac12) - 1)\right)^{k2\lceil\gamma^{-1}\rceil}(1-\eta)^{k}\\
	&= \left(\left(1+\eta(f(\tfrac12) - 1)\right)^{2\lceil\gamma^{-1}\rceil}(1-\eta)\right)^{k}\\
	a^{-T}w_2^{(T)} &= \left(1-\eta\right)^{k\lceil\gamma^{-1}\rceil}\\ &=\left(\left(1-\eta\right)^{\lceil\gamma^{-1}\rceil}\right)^k\\
	\end{align*}
	So it holds that $w_1^{(T)}>w_2^{(T)}$ if we can show that $\left(1+\eta(f(\tfrac12) - 1)\right)^{2\lceil\gamma^{-1}\rceil}(1-\eta)>\left(1-\eta\right)^{\lceil\gamma^{-1}\rceil}$
	
	\begin{align}
	\left(1+\eta(f(\tfrac12) - 1)\right)^{2\lceil\gamma^{-1}\rceil}(1-\eta) &= \left(1-(\tfrac12 - \gamma)\eta\right)^{2\lceil\gamma^{-1}\rceil}(1-\eta)\nonumber &\text{(def. of $\gamma$)}\\
	&\geq \left(1-\eta + 2\gamma\eta\right)^{\lceil\gamma^{-1}\rceil}(1-\eta)\label{e:sym}\\
	&= \left(\frac{1-\eta + 2\gamma\eta}{1-\eta}\right)^{\lceil\gamma^{-1}\rceil}(1-\eta)^{\lceil\gamma^{-1}\rceil+1}\nonumber\\
	&= \left(1+2\gamma
	\eta\right)^{\lceil\gamma^{-1}\rceil}(1-\eta)^{\lceil\gamma^{-1}\rceil+1}\nonumber\\
	&\geq \left(1+\lceil\gamma^{-1}\rceil 2\gamma
	\eta\right)(1-\eta)^{\lceil\gamma^{-1}\rceil+1}\nonumber\\
	&\geq \left(\left(1+2\eta\right)(1-\eta)\right)(1-\eta)^{\lceil\gamma^{-1}\rceil}\nonumber\\
	&>(1-\eta)^{\lceil\gamma^{-1}\rceil} &\text{(for $\eta<\tfrac12$)}\nonumber
	\end{align}
	
	Therefore expert $e_2$ will not incur any more loss during the last stage of the instance, so her total loss is $\mT=k\lceil\gamma^{-1}\rceil$ while the loss of the algorithm is $\MT=T=k\cdot\left(2\lceil\gamma^{-1}\rceil+1\right)$.
	So $$\frac{\MT}{\mt} \geq \frac{k\cdot\left(2\lceil\gamma^{-1}\rceil+1\right)}{k\lceil\gamma^{-1}\rceil} = 2 + \frac{1}{\lceil\gamma^{-1}\rceil}$$
	rearranging proves the claim.
\end{proof}

\subsection{Proof of Lemma~\ref{lem:det-asym2-lb}}
Let $\F$ be a family of scoring rules generated by a normalized
strictly proper scoring rule $f$, with not both $f(0,0)=f(1,1)$ and $f(0,1)=f(1,0)$ and parameters $c$ and $d$ as in Definition~\ref{def:asym}. In the worst case, MW with any scoring rule $f'$ from
$\F$ with $\eta\in(0,\tfrac12)$ can do no better than 
$$\MT \geq\left(2 + \max\{\tfrac{1-c}{2c}, \tfrac{d}{4(1-d)}\}\right)\cdot \mT.$$
\begin{proof}
	Fix $f$, and without loss of generality assume that $f(0,0)=1$ (since $f$ is normalized, either $f(0,0)$ or $f(1,1)$ needs to be $1$, rename if necessary). As $f$ is normalized, at least one of $f(0,1)$ and $f(1,0)$ needs to be $0$. For now, we consider the case where $f(0,1)=0$, we treat the other case later. For now we have $f(0,0)=1$, $f(0,1) = 0$, and by definition~\ref{def:asym}, $f(1,0) = 1-c$ and $f(1,1) = 1-d$, where $c>d$ (since correctly reporting $1$ needs to score higher than reporting $0$ when $1$ materialized) and $\lnot (c=1 \land d=0)$ (since that would put us in the semi-symmetric case).
	
	We construct an instance as follows. We have two experts, $e_0$ reports $0$ always, and $e_1$ reports $1$ always, and as usual, the realizations are opposite of the algorithms decisions. Since the experts have completely opposite predictions, the algorithm will follow whichever expert has the highest weight. We will show that after a constant number of time steps $t$, the weight $w_0^{(t)}$ of $e_0$ will be larger than the weight $w_1^{(t)}$ of $e_1$ even though $e_0$ will have made one more mistake. Note that when this is true for some $t$ independent of $\eta$, this implies that the algorithm cannot do better than $2\tfrac{t}{t-1}>2+\tfrac2t$.
	
	While it hasn't been the case that $w_0^{(t)}>w_1^{(t)}$ with $m_0^{(t)} = m_1^{(t)}+1$, realizations alternate, and the weight of each expert is:
	
	\begin{align}
	w_0^{(2t)} &= a^{2t}(1+\eta(f(0,0)-1))^t(1+\eta(f(0,1)-1))^t\nonumber\\
	&= a^{2t}(1+\eta(1-1))^t(1+\eta(1-c-1))^t\nonumber\\
	&= a^{2t}(1-c\eta)^t\label{e:w0t}\\
	w_1^{(2t)} &= a^{2t}(1+\eta(f(1,1)-1))^t(1+\eta(f(1,0)-1))^t\nonumber\\
	&= a^{2t}(1+\eta((1-d)-1))^t(1+\eta(0-1))^t\nonumber\\
	&= a^{2t}(1-d\eta)^t(1-\eta)^t\label{e:w1t}
	\end{align}
	
	What remains to be shown is that for some $t$ independent of $\eta$, $$a^{2t+1}(1-c\eta)^{t+1}>a^{2t+1}(1-d\eta)^{t+1}(1-\eta)^t.$$
	
	We know that it cannot be the case that simultaneously $c=1$ and $d=0$, so let's first consider the case where $c<1$. In this case, it is sufficient to prove the above statement assuming $d=0$, as this implies the inequality for all $d\in[0,c)$. The following derivation shows that $a^{2t+1}(1-c\eta)^{t+1} > a^{2t+1}(1-d\eta)^{t+1}(1-\eta)^t$ whenever $\frac{c}{(1-c)} < t$.
	
	\begin{align*}
	a^{2t+1}(1-c\eta)^{t+1} &> a^{2t+1}(1-d\eta)^{t+1}(1-\eta)^t\\
	(1-c\eta)^{t+1} &> (1-\eta)^t & \text{($d=0$)}\\
	(1-c\eta) &> \left(\frac{1-\eta}{1-c\eta}\right)^t\\
	\ln(1-c\eta) &> t\cdot\ln\left(\frac{1-\eta}{1-c\eta}\right)\\
	1-\frac{1}{1-c\eta} &> t\cdot\left(\frac{1-\eta}{1-c\eta}-1\right) & (1-\tfrac1x \leq \ln x \leq x-1)\\
	\frac{1-c\eta-1}{1-c\eta} &> t\cdot\left(\frac{1-\eta-1+c\eta}{1-c\eta}\right)\\
	\frac{c\eta}{1-c\eta} &< t\cdot\frac{(1-c)\eta}{1-c\eta}\\
	\frac{c\eta}{(1-c)\eta} &< t\\
	\frac{c}{(1-c)} &< t\\
	\end{align*}
	
	So after $2t+1$ time steps for some $t\leq \frac{c}{1-c}+1$, expert $e_0$ will have one more mistake than expert $e_1$, but still have a higher weight. This means that after at most another $2t+1$ time steps, she will have two more mistakes, yet still a higher weight. In general, the total loss of the algorithm is at least $2+\frac{1-c}{c}$ times that of the best expert. Now consider the case where $c=1$ and therefore $d>0$. We will show that after $2t+1$ time steps for some $t\leq 2\frac{1-d}{d}+1$ expert $e_0$ will have one more mistake than expert $e_1$.
	
	\begin{align*}
	a^{2t}(1-c\eta)^{t+1} &> a^{2t}(1-d\eta)^t(1-\eta)^t(1-d\eta)\\
	(1-\eta)^{t+1} &> (1-d\eta)^{t+1}(1-\eta)^t & \text{($c=1$)}\\
	\frac{1-\eta}{1-d\eta} &> (1-d\eta)^t\\
	\ln\left(\frac{1-\eta}{1-d\eta}\right) &> t\ln(1-d\eta)\\
	1-\frac{1-d\eta}{1-\eta} &> t(1-d\eta-1) & (1-\tfrac1x \leq \ln x \leq x-1)\\
	\frac{1-\eta-1+d\eta}{1-\eta} &> -td\eta\\
	\frac{(1-d)\eta}{1-\eta} &< td\eta\\
	\frac{1-d}d\frac1{1-\eta} &< t\\
	2\frac{1-d}d &< t & (\text{by }\eta<\tfrac12)
	\end{align*}
	
	So in any case, after $t \leq 2\max\{\tfrac{c}{1-c}, \tfrac{1-d}{d}\} +1$ time steps so the loss compared to the best expert is at least $$2+\max\{\tfrac{1-c}{c}, \tfrac{d}{2(1-d)}\}.$$
	
	What remains to be proven is the case where $f(0,1)>0$. In this case, it will have to be that $f(1,0)=0$, as $f$ is normalized. And similarly to before, by Definition~\ref{def:asym}, we have $f(0,1) = 1-c$ and $f(1,1)=1-d$ for $c>d$ and $\lnot(c=1 \land d=0)$. Now, whenever $w_o^{(t)}>w_1{(t)}$, $e_0$ will predict $1$ and $e_1$ predicts $0$, and otherwise $e_0$ predicts $0$ and $e_1$ predicts $1$. As usual, the realizations are opposite of the algorithm's decisions. For now assume tie of the algorithm is broken in favor of $e_1$, then the weights will be identical to \eqref{e:w0t}, \eqref{e:w1t}. If the tie is broken in favor of $e_0$ initially, it takes at most twice as long before $e_0$ makes two mistakes in a row. Therefore, the loss with respect to the best expert in hindsight of an algorithm with any asymmetric strictly proper scoring rule is $$2+\max\{\tfrac{1-c}{2c}, \tfrac{d}{4(1-d)}\}.$$
\end{proof}

\subsection{Proof of Theorem~\ref{thm:det-non-ic-lb-continuous}}
	For a weight update function $f$ with continuous strictly
increasing rationality function $\rho_f$, with
$\rho_f(0)<\tfrac12<\rho_f(1)$ and $\rho_f(\tfrac12) =
\tfrac12$, there is no deterministic no $2$-regret algorithm.

\begin{proof}
	Fix $f$ with $\rho_f(0)<\tfrac12<\rho_f(1)$ and $\rho_f(\tfrac12) = \tfrac12$. Define $p = \max\{\rho_f(0), 1-\rho_f(1)\}$, so that $p$ and $1-p$ are both in the image of $\rho_f$ and the difference between $p$ and $1-p$ is as large as possible. Let $b_1 = \rho^{-1}(p)$ and $b_2=\rho^{-1}(1-p)$ and observe that $b_1 < \tfrac12 < b_2$.
	
	Next, we rewrite the weight-update function $f$ in a similar way as the normalization procedure similar to Definition~\ref{def:family}:
	$f(p,r) = a(1+\eta(f'(p,r) -1)).$ where $\max\{f'(p,0), f'(1-p,1)\}=1$ and $\min\{f'(p,1), f'(1-p,0)\}=0$. Again we do this to prove bounds that are not dependent on any learning rate parameter.
	
	Note that the composition of $\rho_f$ and $f$, namely $f(\rho_f(p),r)$ is a strictly proper scoring rule, since it changes the prediction in the same way as the selfish expert would do. Since $f(\rho_f(p),r)$, it must also be that $f'(\rho_f(p),r)$ is a strictly proper scoring rule, since it is a positive affine transformation of $f\circ \rho_f$.\footnote{And since $f'$ is a positive affine transformation of $f$, the rationality function is unchanged due to the linearity of the expectation operator.}
	
	We now continue similarly to the lower bounds in Section~\ref{s:det-lb}. We only treat the semi-symmetric and asymmetric cases as the former includes the special case of the symmetric weight-update function.
	
	For the semi-symmetric case, by definition $f'(\rho_f(b_1),0) = f'(\rho_f(b_2),1)=1$ and $f'(p,0), f'(1-p,1)\}=1$ and $\min\{f'(p,1), f'(1-p,0)=0$. Because $f'\circ\rho_f$ is a strictly proper scoring rule, the following inequality holds:
	
	\begin{align*}
	\tfrac12 f'(\rho_f(\tfrac12), 0) + \tfrac12 f'(\rho_f(\tfrac12,1)) +\mu = \tfrac12 f'(\rho_f(b_1), 0) +  \tfrac12 f'(\rho_f(\tfrac12,1)) = \tfrac12
	\end{align*}
	for some $\mu>0$, since an expert with belief $\rho_f(\tfrac12)$ must have a strict incentive to report this. Here $\mu$ plays the same role as the semi-symmetric scoring rule gap.\footnote{It is defined slightly differently though, as the image of $\rho_f$ may not be $[0,1]$.}
	
	We now pitch three experts against each other in a similar lower bound instance as Lemma~\ref{lem:det-asym1-lb}. For the first stage, they have beliefs $b_0^{(t)}=\tfrac12$, $b_1^{(t)}=b_1$, $b_2^{(t)}=b_2$, so they have predictions $p_0^{(t)}=\tfrac12$, $p_1^{(t)}=\rho_f(b_1)=p$, $p_2^{(t)}=\rho_f(b_2)=1-p$. For the second stage, recall that either $b_1=0$ or $b_2=1$. In the former case, $b_0^{(t)}=1$, $b_1^{(t)}=0$, $b_2^{(t)}=\tfrac12$ and $\rt=0$ and in the latter case $b_0^{(t)}=0$, $b_1^{(t)}=1$, $b_2^{(t)}=\tfrac12$ and $\rt=1$. We now show a bijection between the instance in Lemma~\ref{lem:det-asym1-lb} and this instance, which establishes the lower bound for the semi-symmetric non-incentive compatible case. First of all, note that the weights of each of the experts in the first stage is the same (up to the choice of $a$ and $\eta$, and for now assuming that the algorithms choices and thus the realizations are the same):
	\begin{align*}
	w_0^{(2t)} &= a^{2t}\left((1+\eta(f'(\tfrac12,0)-1))((1+\eta(f'(\tfrac12,1)-1))\right)^t\\
	&\geq a^{2t}(1-\eta+2\mu\eta)^t&(\text{Follows from \eqref{e:semisym}})\\
	w_1^{(2t)} &=a^{2t}\left(1-\eta)\right)^t\\
	w_1^{(2t)} &=a^{2t}\left(1-\eta)\right)^t
	\end{align*}
	In the second stage expert $e_0$ is always wrong and $e_1$ is always right, and hence at time $T$ the weights
	
	Also note, that the predictions of $e_1$ and $e_2$ are opposite, i.e. $p$ and $1-p$, so the algorithm will follow the expert which highest weight, meaning the algorithms decisions and the realizations are identical to the instance in Lemma~\ref{lem:det-asym1-lb}.
	
	To complete the proof of the lower bound instance, we need to show that the total loss of $e_1$ is the same. During the first stage, alternatingly the true absolute loss of $e_1$ is $b_1$ and $1-b_1$, so after each 2 steps, her loss is $1$. During the last stage, since her belief is certain (i.e. $b_0$ if $b_1=0$ or $b_2$ if $b_2=1$) ans she is correct, she incurs no additional loss. Therefore the loss of the algorithm and the true loss of $e_1$ are the same as in Lemma~\ref{lem:det-asym1-lb}, hence the loss of the algorithm is at least $\frac{1}{\lceil\mu^{-1}\rceil}$ times that of the best expert in hindsight.
	
	Finally, we consider the asymmetric case. We use a similar instance as Lemma~\ref{lem:det-asym2-lb} with two experts $e_0$, $e_1$. If $f'(1-p,0)=0$ we have $b_0^{(t)}=b_1$ and $b_1^{(t)}=b_2$, so $p_0^{(t)}=p$ and $p_1^{(t)} = 1-p$, otherwise the beliefs (and thus predictions) alternate. Again, the predictions are opposite of each other, and the weights evolve identically (up to the choice of $a$ and $\eta$) as before. Again the loss up until the moment that the same expert is chosen twice in a row is the same. 
	
	Once the same expert is chosen twice (after at most $2\max\{\tfrac{c}{1-c}, \tfrac{1-d}{d}\} +1)$ steps), it is not necessarily the case that the total loss of one expert exceeds the other by 2, as the true beliefs are $b_1$ and $b_2$, rather than $0$ and $1$. However, since at least either $b_1=0$ or $b_2=1$, and $b_1<\tfrac12<b_2$, the difference in total absolute loss in this non-IC instance is at least half of the IC instance, so we lose at most  factor $\tfrac12$ in the regret bound, hence for the asymmetric case $\MT \geq \left(2+\max\{\tfrac{1-c}{4c}, \tfrac{d}{8(1-d)}\}\right)\mt$, completing the proof of the statement.
\end{proof}

\subsection{Proof of Theorem~\ref{thm:rand-ub}}
Let $A\in\A$ be a $\theta$-RWM algorithm with the Brier weight update
rule $f_\text{Br}$ and $\theta=0.382$ and with
$\eta\in(0,\tfrac12)$. For any expert $i$ it holds that $$\MT \leq
2.62\left((1+\eta)\mT + \frac{\ln n}{\eta}\right).$$ 
\begin{proof}
	The core difference between the proof of this statement, and the proof
	for Theorem~\ref{thm:spherical-det-ub} is in giving the upper bound of
	$\Phi^{(t+1)}$. Here we will give an upper bound of $\Phi^{(T)} \leq
	n\cdot \exp\left(-\frac{\eta}{2.62}\MT\right)$. Before giving this
	bound, observe that this would imply the theorem: since the weight
	updates are identical to the deterministic algorithm, we can use the
	same lower bound for $\Phi^{(T)}$, namely $\Phi^{(T)}\geq
	(1-\eta)^{\mT}$ for each expert. Then taking the log of both sides we
	get: 
	
	\begin{align*}
	\ln n - \frac{\eta}{2.62}\MT &\geq \mT\cdot\ln(1-\eta)\\
	\ln n - \frac{\eta}{2.62}\MT &\geq \mT\cdot(-\eta-\eta^2)\\
	\MT &\leq 2.62\left((1+\eta)\mT + \frac{\ln n}{\eta}\right)
	\end{align*}
	
	So all that's left to prove is that whenever the algorithm incurs a
	loss $\ell$, $\Phi^{(t+1)}\leq
	\exp\left(-\tfrac{\eta}{2.62}\ell\right)$. At time $t$, the output
	$\qt$ of a $\theta$-RWM algorithm is one of three cases, depending on
	the weighted expert prediction. The first options is that the
	algorithm reported the realized event, in which case the
	$\ell^{(t)}=0$ and the statement holds trivially. We treat the other
	two cases separately. 
	
	Let's first consider the case where the algorithm reported the
	incorrect event with certainty: $\ell^{(t)}=1$. The means that
	$\sum_{i=1}^n \wt\st \geq (1-\theta)\Phi^{(t)}$. Since the Brier rule
	is concave, $\Phi^{(t+1)}$ is maximized when $\st = 1-\theta$ for all
	experts $i$. In this case each we get 
	
	\begin{align*}
	\Ptp &\leq \sum_i\left(1-\eta\left(\frac{(\pt)^2 + (1-\pt)^2 +1}{2}-(1-\st) \right)\right)\wt\\
	&\leq \sum_i\left(1-\eta\left(\frac{(\theta)^2 + (1-\theta)^2 +1}{2}-\theta \right)\right)\wt\\
	&\leq\sum_i\left(1-\frac{\eta}{2.62}\right)\wt &\text{(since $\theta = .382$)}\\
	&=\left(1-\frac{\eta}{2.62}\ell^{(t)}\right)\Pt.
	\end{align*}

	Otherwise the algorithms report is between $\theta$ and $1-\theta$. Let $\ell^{(t)}\in[\theta,1-\theta]$ be the loss of the algorithm. Again, since the Brier rule is concave, $\Phi^{(t+1)}$ is maximized when $\st = \ell^{(t)}$ for all experts $i$. On $[\theta, 1-\theta]$ the Brier proper scoring rule can be upper bounded by
	$$1-\frac{\eta}{f_\text{Br}(1-\theta,1)/\theta} \st = 1-\frac{\eta}{2.62}\st.$$
	
	This yields
	\begin{align*}
	\Ptp &\leq \sum_i\left(1-\eta\left(\frac{(\pt)^2 + (1-\pt)^2 +1}{2}-(1-\st) \right)\right)\wt\\
	&\leq \sum_i\left(1-\frac{\eta}{2.62}\st \right)\wt\\
	&\leq \left(1-\frac{\eta}{2.62}\ell^{(t)} \right)\Pt
	\end{align*}
	
	So the potential at time $T$ can be bounded by $\Phi^{(T)} \leq n\cdot\prod_t \left(1-\frac{\eta}{2.62}\ell^{(t)}\right)\leq n\cdot\exp\left(-\frac{\eta}{2.62}\MT\right)$, from which the claim follows.
\end{proof}

\section{Selecting a Strictly Proper Scoring Rule}
\label{s:selection}

\begin{figure}
	\centering
	\includegraphics[width=.6\textwidth]{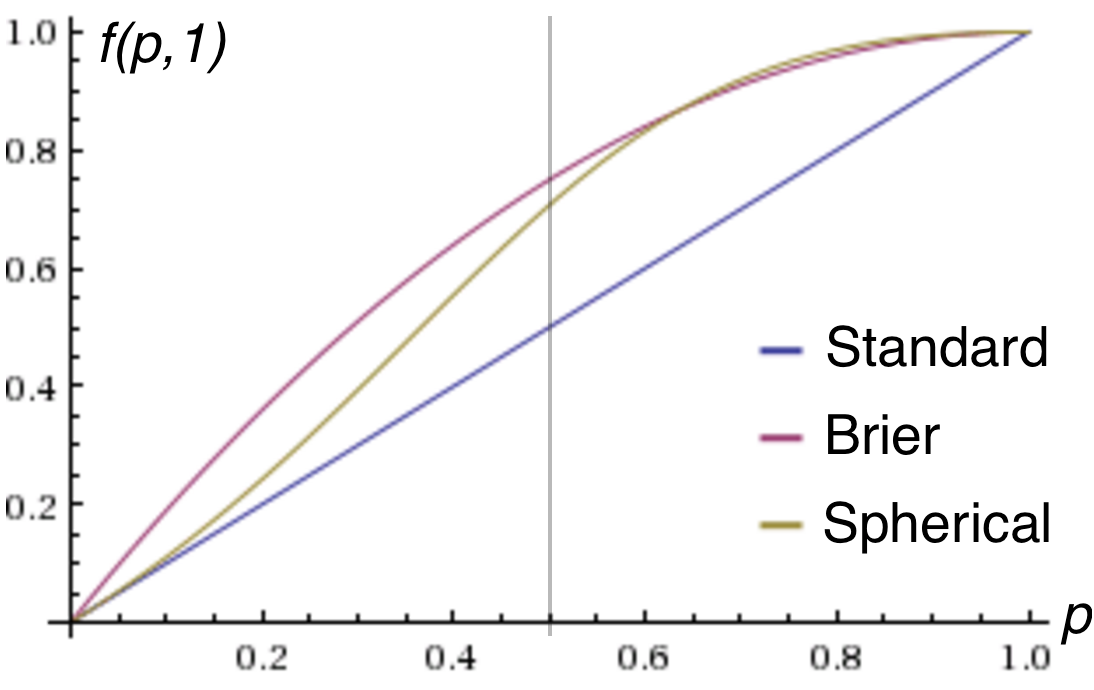}
	\caption{Three different normalized weight-update rules for $\rt=1$. The line segment is the standard update rule, the concave curve the Brier rule and the other curve the spherical rule.}
	\label{f:spsr}
\end{figure}

When selecting a strictly proper scoring rule for an IC online prediction algorithm, different choices may lead to very different guarantees. Many different scoring rules exist \citep{M56,S71}, and for discussion of selecting proper scoring rules in non-online settings, see also \citep{MS13}. Figure~\ref{f:spsr} shows two popular strictly proper scoring rules, the Brier rule and the spherical rule, along with the standard rule as comparison. Note that we have normalized all three rules for easy comparison.

Firstly, we know that for honest experts, the standard rule performs
close to optimally. For every $\delta>0$ we can pick a learning rate
$\eta$ such that as $T\rightarrow \infty$ the loss of the algorithm
$\MT \leq (2+\delta)\mt$, while no algorithm could do better than $\MT
< 2\mT$ \citep{LW94,FS97}. This motivates looking at strictly proper
scoring rule that are ``close'' to the standard update rule in some sense.
In Figure~\ref{f:spsr}, if we compare the two
strictly proper scoring rules, the spherical rule seems to follow the
standard rule better than Brier does. 

A more formal way of look at this is to look at the scoring rule
gap. In Figure~\ref{f:spsr} we marked the $p =\tfrac12$
location. Visually, the scoring rule gap $\gamma$ is the difference
between a scoring rule and the standard rule at $p=\tfrac12$. Since
the Brier score has a large scoring rule gap, we're able to prove a
strictly stronger lower bound for it: the scoring rule gap $\gamma =
\tfrac14$, hence MW with the Brier scoring rule cannot do better than
$\MT\geq (2+\tfrac14)\mT$ in the worst case, according to
Lemma~\ref{lem:det-sym-lb}. Corollary~\ref{cor:spherical} shows that
for the Spherical rule, this factor is $2+\tfrac15$. The ability to
prove stronger lower bounds for scoring rules with larger gap
parameter $\gamma$ is an indication that it is probably harder to
prove strong upper bounds for those scoring rules.


\end{document}